\documentclass{article}

\usepackage{cooldraft}

\usepackage[utf8]{inputenc}
\usepackage[T1]{fontenc}
\usepackage{hyperref}
\usepackage{url}
\usepackage{booktabs}
\usepackage{amsfonts}
\usepackage{nicefrac}
\usepackage{microtype}
\usepackage{xcolor}
\usepackage{xargs}
\usepackage{amsmath}
\usepackage{amsthm}
\usepackage{amssymb}
\usepackage{tikz}
\usetikzlibrary{arrows, shapes.geometric, automata, positioning}
\usepackage{thmtools}
\usepackage{thm-restate}
\usepackage{subcaption}
\usepackage[b]{esvect}
\usepackage[inline]{enumitem}

\hypersetup{
  colorlinks=true,
  linkcolor=blue!60!black,
  citecolor=blue!60!black,
  urlcolor=blue!60!black
}

\theoremstyle{plain}
\newtheorem{theorem}{Theorem}
\newtheorem{lemma}[theorem]{Lemma}
\newtheorem{proposition}[theorem]{Proposition}
\newtheorem{corollary}[theorem]{Corollary}

\theoremstyle{definition}
\newtheorem{definition}[theorem]{Definition}
\newtheorem{problem}[theorem]{Problem}
\newtheorem{remark}[theorem]{Remark}
\newtheorem{example}[theorem]{Example}

\newcommand{\nat}{\mathbb{N}}

\newcommand{\real}{\mathbb{R}}

\newcommand{\distr}{\mathrm{Distr}}
\newcommand{\prob}{\mathbb{P}}

\newcommand{\rand}[1]{\boldsymbol{#1}}
\newcommand{\supp}{\mathrm{supp}}
\newcommand{\one}{\mathbf{1}}

\newcommand{\diag}{\operatorname{diag}}

\newcommand{\krank}{\mathrm{krank}}
\newcommand{\linspan}{\operatorname{span}}
\newcommand{\transpose}{^{\intercal}}
\newcommand{\rowtensor}{\otimes^{\mathrm{row}}}

\makeatletter
\DeclareMathAlphabet{\mathbb@bbold}{U}{bbold}{m}{n}
\newcommand{\indicator}{\mathbb@bbold{1}}
\makeatother

\newcommand{\model}{\mathcal{M}}
\newcommand{\trans}{A}
\newcommand{\obs}{B}
\newcommand{\init}{\pi}

\newcommand{\hmms}{\mathcal{H}}
\newcommand{\hmmsn}[1]{\mathcal{H}_{#1}}
\newcommand{\abc}{\mathrm{\Sigma}}

\newcommand{\similarperm}{\cong}
\newcommand{\othereq}{\bowtie}
\newcommand{\target}{{\model^\star}}
\newcommand{\rival}{{\model'}}
\newcommand{\masked}{\mathcal{U}}

\newcommand{\iddeteq}{\mathsf{IDHMM}_\mathrm{det}^{=}}

\newcommand{\iddetperm}{\mathsf{IDHMM}_\mathrm{det}^\mathrm{perm}}

\newcommand{\iddetbisim}{\mathsf{IDHMM}_\mathrm{det}^\sim}
\newcommand{\idlocal}{\mathsf{IDHMM}_\mathrm{local}}
\newcommand{\idgeneric}{\mathsf{IDHMM}_\mathrm{gen}}
\newcommand{\iddetfam}{\mathsf{IDHMM}_\mathrm{fam}}
\newcommand{\biniddeteq}{\mathsf{BIN}\textnormal{-}\iddeteq}

\newcommand{\coiddeteq}{\mathsf{coIDHMM}_\mathrm{det}^{=}}

\newcommand{\coiddetperm}{\mathsf{coIDHMM}_\mathrm{det}^\mathrm{perm}}
\newcommand{\coiddetbisim}{\mathsf{coIDHMM}_\mathrm{det}^\sim}
\newcommand{\coidlocal}{\mathsf{coIDHMM}_\mathrm{local}}

\newcommand{\np}{\mathsf{NP}}
\newcommand{\conp}{\mathsf{coNP}}
\newcommand{\gi}{\mathsf{GI}}
\newcommand{\pspace}{\mathsf{PSPACE}}

\newcommand{\etr}{\exists \real}
\newcommand{\quadetr}{\mathsf{QUAD}}
\newcommand{\coetr}{\forall \real}

\newcommand{\cmark}{\scalebox{1.0}{\checkmark}}
\newcommand{\xlength}{0.55}
\newcommand{\xsymb}{%
  \tikz[baseline=-0.55ex, x=1ex, y=1ex]{
    \draw[-, line width=0.7pt, line cap=round]
      (-\xlength,-\xlength) -- (\xlength,\xlength)
      (-\xlength,\xlength) -- (\xlength,-\xlength);
  }%
}
\newcommand{\xmark}{\hspace{0.5pt}\xsymb\hspace{0.5pt}}

\title{On the Computational Complexity of Hidden Markov Model Identification}

\author{
  Markel Zubia\\
  Ruhr University Bochum\\
  \texttt{markel.zubia@rub.de}
  \And
  Nils Jansen\\
  Ruhr University Bochum\\
  Radboud University Nijmegen
}

\begin{document}

\maketitle

\begin{abstract}
Identification is the task of recovering the parameters of an unknown ground-truth model from sampled data.
When parameters other than the ground truth induce the same output distribution, data alone does not provide enough information to recover the ground truth, and the model is thus called \emph{unidentifiable}.
We study the \emph{identifiability problem} for hidden Markov models (HMMs): given an HMM, is it identifiable?
Existing work on HMM identification establishes conditions under which the ground-truth HMM can be identified.
However, most of these conditions are sufficient but not necessary, meaning that, when a model does not satisfy them, its identifiability remains inconclusive.
We instead take a computational perspective: is there a sound and complete algorithm that decides whether a given HMM is identifiable, and if so, what is the complexity of this decision problem?
We consider the decision problems arising from the various notions of identifiability in the literature, including deterministic, generic, global, local, state-permutation-invariant, and finite-alphabet identifiability.
We show that all of these problems are decidable in $\pspace$, via reductions to the theory of the reals at various levels of its quantifier-alternation hierarchy.
We further show that the deterministic variants are already $\coetr$-hard (and hence $\conp$-hard) for simply parameterized families.
\end{abstract}

\section{Introduction}

\paragraph{Hidden Markov models.}
A Hidden Markov model~\citep[HMM;][]{petrie1969probabilistic,rabiner1989tutorial} is a stochastic model in which an unobserved Markov chain generates sequences of observations.
Given a dataset of observation sequences, the central task is HMM \emph{estimation}: find an HMM that generates the data.
The typical approach to HMM estimation is maximum-likelihood estimation (i.e., estimate the likeliest HMM to generate the data), which is often implemented by the famed Baum–Welch algorithm~\citep{baum1970maximization,leroux1992maximum}.
As such, HMMs are a standard tool for modeling sequential data and play a crucial role in a wide range of fields, including signal processing~\citep{rabiner1989tutorial}, machine learning~\citep{ghahramani2001introduction}, finance~\citep{hamilton1989new}, psychology~\citep{visser2002fitting}, telecommunications~\citep{turin1998hidden}, and many more~\citep{rabiner1989tutorial,ephraim2002hidden}.

\paragraph{HMM identification.}
Within the vast landscape of HMM estimation, a considerable body of work focuses on \emph{identification}~\citep{blackwell1957identifiability,petrie1969probabilistic,rothenberg1971identification}: under the assumption that the dataset was generated by a ground-truth HMM, beyond simply fitting a model to said data, can we reconstruct the ground-truth HMM? 
If there exists an HMM different from the ground truth that generates the same distribution on observation sequences, then no amount of data from that distribution can allow us to distinguish between them, and the ground truth is thus deemed \emph{unidentifiable}. 
In contraposition, the ground-truth is \emph{identifiable} when its distribution on observation sequences is \emph{unique}.
HMM identifiability is particularly relevant in settings where the ground-truth model is an interesting object of study in-and-of-itself, having been applied in a broad range of areas such as cognitive diagnosis \citep{liu2023identifiability}, system monitoring~\citep{zhang2016bayesian}, and DNA sequencing~\citep{steel1993invertible,steel1998reconstructing}.

\paragraph{Identifiability flavors.}
For any ground-truth HMM, one may construct a different HMM with the same output probabilities in a trivial manner, say by adding dummy hidden states or by renaming the hidden states without changing the output probabilities.
Therefore, to pose a sensible notion of identifiability, it is necessary to impose restrictions on the set of allowed HMMs, for instance by fixing the number of hidden states, as well as specify the kind of uniqueness we seek, for instance by enforcing uniqueness up to hidden-state renaming.
Consequently, literature introduces various nuanced definitions of HMM identifiability.
There are two main flavors of identifiability: \emph{deterministic} and \emph{generic}.
Deterministic (or strict) identification is a classical idea~(see \citet{paz1971introduction,finesso1990consistent,bonhomme2016estimating}), and it expresses that the ground-truth model is ``strictly'' unique, in that, no model other than the ground-truth may generate the data.
In contrast, generic (or probabilistic) identification asks whether the HMMs that induce the same distribution on output words as the ground-truth model form a set of Lebesgue measure zero (see \citet{petrie1969probabilistic,allman2009identifiability,liu2023identifiability}), which in principle allows the identification of the ground-truth almost-surely.
Literature further establishes additional variants, and we refer to \citet{cole2020parameter} for a comprehensive introduction.
In short, \emph{local} identification asks whether there exists a ball around the ground-truth in which it is unique (whereas the more general question is regarded as \emph{global} identification), \emph{parameterized} identification restricts the number of parameters that describe the model (closely related to \emph{parameter redundancy} \citep{catchpole1997detecting}), \emph{permutation independence} considers uniqueness up to the permutation of hidden states. 
Finally, some works use the term ``model'' to refer to individual HMMs~\citep{tune2013hidden}, whereas others use it to refer to classes or sets of HMMs induced by a given set of parameters~\citep{catchpole1997detecting}.
We adopt the former terminology, and we explicitly write ``family'' when referring to sets of HMMs.

\paragraph{Existing HMM identification techniques and the current gap.}
Most HMM literature on identifiability studies the conditions under which identification is guaranteed.
These conditions include rank and separation assumptions used by spectral algorithms \citep{DBLP:journals/jcss/HsuKZ12}, 
assumptions on Markov chain topology \citep{pearl1986structuring,steel1998reconstructing,chang1996full}, 
Kruskal-rank conditions \citep{kruskal1976more,allman2009identifiability},
and problem-specific restrictions such as Q-matrix structures in cognitive diagnosis \citep{liu2023identifiability}.
\citet{tune2013hidden} build upon \citet{allman2009identifiability} to provide necessary and sufficient conditions for identifiability under irreducibility, aperiodicity, and non-redundancy assumptions (see Appendix~\ref{app:tnrcondition}).
While these results are crucial, our work presents a different perspective: rather than introducing another criterion for identifiability, we focus on the computational complexity of the decision problem itself.
\citet{cole2019parameter,cole2020parameter} applies symbolic-algebra systems to local HMM identifiability, which provides a sound and complete decision algorithm. 
However, a naive application of such approach into our setting requires exhaustively enumerating exponentially many probabilities, whereas we provide a $\pspace$ solution.
Even though identifiability problems are often computationally hard for other stochastic models (see for instance \citet{van2023hardness,dorfler2024complexity}), and HMM identifiability has vast literature and has generally been speculated intractable~\citep{cole2019parameter}, the exact computational complexity of HMM identifiability has so far remained an open question.

\paragraph{Scope.}
We focus on finite-state, finite-alphabet, discrete HMMs.
This is the setting of much of the classical HMM identifiability and probabilistic automata literature~\citep{petrie1969probabilistic,baum1970maximization,paz1971introduction,rabiner1989tutorial,tzeng1992polynomial}.
It is also the natural setting from a computational perspective: the input is an HMM with rational parameters that can be encoded in the tape of a Turing machine.
Other models such as continuous-output, infinite state, or nonparametric HMMs are important statistical models, but their identifiability is not as relevant from a computational perspective and are not treated in this paper.

\paragraph{Research objectives.}
We study the (Cook–Levin–Karp) computational complexity-theoretic classification of canonical HMM identifiability problems, by proving complexity upper bounds in the form of decision algorithms (i.e., it cannot be harder than running its decision algorithm), and lower bounds in the form of many-one reductions (i.e., it is at least as hard as another well-known hard problem). 
In particular, we study the following problems:
\begin{enumerate*}[label=(\Roman*)]
    \item deterministic HMM identifiability up to equality, hidden state permutation, and bisimilarity,
    \item deterministic identifiability with fixed alphabet sizes,
    \item local identifiability,
    \item generic identifiability, and
    \item family-wise identifiability.
\end{enumerate*}

\paragraph{Contribution.}
We give polynomial-size $\etr$ certificates for HMM output equivalence in Section~\ref{sec:hmm_equivalence_certificate} and use them throughout the paper.
Section~\ref{sec:deterministic_hmm_equality} shows that deterministic identifiability up to equality is in $\coetr$, and that it is $\coetr$-complete for simple parameterizations, and $\conp$-hard for mask (rectangular) parameterizations.
We show that identifiability up to hidden-state permutation is in $\forall \etr$ in Section~\ref{sec:permutation}, while identifiability up to bisimilarity is in $\coetr$ is shown in Section~\ref{sec:bisimilarity}.
Section~\ref{sec:perm_bisim_hard} shows that the hardness results from deterministic global identifiability up to equality carry over to both identifiability up to hidden-state renaming and bisimilarity.
In Section~\ref{sec:binary-alphabet}, we show that the hardness results are alphabet-independent, meaning that the identification of HMMs with a binary alphabet is as hard as the problem with an arbitrary alphabet.
Section~\ref{sec:local_identifiability} shows that local identifiability is in $\coetr$. 
Section~\ref{sec:generic_identifiability} shows that generic identifiability is in $\exists\coetr \cap \forall \etr$, and Section~\ref{sec:family-wise} shows family-wise identifiability is in $\coetr$.
By \citet{canny1988some} and \citet{renegar1992computational}, it follows that all problems studied here are decidable in $\pspace$.

\section{Related work}

\paragraph{HMM learning.}
Learning HMMs from data is generally computationally difficult
\citep{abe1992computational,terwijn2002learnability}.  The Baum--Welch algorithm
and related expectation-maximization methods remain standard practical tools
\citep{baum1970maximization,baum1972inequality,rabiner1989tutorial}, while Bayesian and approximate
inference methods are also widely used~\citep{stolcke1992hidden}.  Positive
polynomial-time learning guarantees are known under conditions such as
rank or separation assumptions~\citep{DBLP:journals/jcss/HsuKZ12,mossel2005learning}.

\paragraph{Deterministic HMM identification.}
The deterministic (or strict) identifiability problem for functions of finite Markov chains goes back to
\citet{blackwell1957identifiability} and \citet{petrie1969probabilistic}.
Later work developed sufficient conditions for deterministic identifiability, including
conditions based on consistency of estimation~\citep{finesso1990consistent},
multivariate latent-structure decompositions~\citep{bonhomme2016estimating}, and
nonparametric HMM inference~\citep{gassiat2016inference}.
\citet{tune2013hidden} use tensor decompositions to give a deterministic identifiability criterion for finite-alphabet HMMs under irreducibility and aperiodicity assumptions. 
We view their theorem as an edge-case and provide a discussion in Appendix~\ref{app:tnrcondition}. 

\paragraph{Generic HMM identification.}
Generic identifiability asks whether the non-identifiable parameters form a set with a lower algebraic dimension.
\citet{allman2009identifiability} establish generic identifiability for broad latent-structure models, including HMMs, using Kruskal's theorem.
Related generic identifiability criteria appear in work on multivariate latent structures and nonparametric HMMs~\citep{bonhomme2016estimating,gassiat2016inference}.
Our generic-identifiability section does not propose another criterion for generic identifiability, but studies the computational complexity of its decision problem.

\paragraph{Local identifiability and parameter redundancy.}
A common way to study local identifiability is through \emph{parameter redundancy}: whether a parameterized statistical model uses more parameters than can be independently recovered from its observable probabilities~\citep{rothenberg1971identification,catchpole1997detecting}.
Cole applies this idea to HMMs using symbolic algebra.
Her calculation first forms an \emph{exhaustive summary}, namely a finite list of expressions in the parameters that fully describe the model's output distribution.
One then differentiates these expressions with respect to the parameters and studies the rank of the resulting matrix~\citep{cole2019parameter,cole2020parameter}.
Cole uses this rank to decide whether the parameterized family is redundant or at least locally identifiable.

\paragraph{Hardness of identification in other models.}
\citet{dorfler2024complexity} study generic identification in linear structural causal models.
They prove $\coetr$-hardness for numerical identification and give $\pspace$ upper bounds for generic identification.
Their lower-bound strategy is based on ``finding another solution'' promise variants such as $\mathsf{QUAD}^{++}$, and is structurally close to HMM non-identifiability: in both settings, one asks whether a known model is unique.

\paragraph{Automaton equivalence.}
Probabilistic and weighted automata provide a useful bridge between HMM learning and automata theory.
A weighted automaton takes a sequence of symbols (a word) as input and returns a real number (its weight) as output by assigning a real-valued matrix to every symbol and taking the product of the symbols that appear in the input word.
Probabilistic automata are the stochastic specialization of weighted automata~\citet{paz1971introduction}.
HMM word probabilities have a similar matrix-product form, which is why algorithms for probabilistic automaton equivalence can be reused for HMMs~\citep{doyen2008equivalence}.
In particular, the probabilistic automaton equivalence theorems by \citet{paz1971introduction}, \citet{tzeng1992polynomial}, and \citet{kiefer2014stability} form the basis of the output-equivalence certificate in Section~\ref{sec:hmm_equivalence_certificate}.

\paragraph{Automaton minimization.}
Automaton minimization asks whether there exists an automaton with fewer states than a given automaton, while performing the same function.
For weighted automata this can be done in polynomial time by constructing the basis of the vectors reachable from the initial vector and, symmetrically, of the vectors that can reach the final vector~\citep{schutzenberger1961definition,berstel1988rational,kiefer2014stability}.
These linear-algebraic approaches for weighted automaton minimization do not necessarily preserve nonnegativity and stochasticity, so they cannot be used for probabilistic automaton minimization.
\citet{kiefer2014stability} show that probabilistic-automaton minimization is NP-hard and reduce it in polynomial time to the existential theory of the reals.
The same basis-of-reachable-vectors idea will reappear in our certificate for HMM output equivalence in Section~\ref{sec:hmm_equivalence_certificate}.

\paragraph{Automata learning.}
Active automata learning begins with Angluin's $L^*$ algorithm for regular languages~\citep{angluin1987learning}, which has been extended to probabilistic models such as Markov chains (MC) and Markov decision processes (MDP)~\citep{tappler2021based}.
Deterministic automaton identification from labeled examples is already computationally hard: 
\citet{gold1978complexity} proved that finding a smallest DFA consistent with given data is NP-complete. 
This literature is distinct from our setting, which assumes access to the full output distribution rather than labeled examples, but it provides useful context for why automaton identification problems often become computationally difficult.

\section{Background}

\subsection{Notation}

A probability distribution over a finite set $X$ is a function $\mu \colon X \to [0,1]$ where $\sum_{x\in X}\mu(x) = 1$.
The support of $\mu$ is $\supp(\mu) := \{x\in X\mid \mu(x)>0\}$.
For a kernel $\sigma\colon X\to\distr(Y)$, we write $\sigma(y\mid x) := \sigma(x)(y)$.
We use boldface for random variables.
Indicator $\indicator_\varphi$ is $1$ when $\varphi$ holds and $0$ otherwise.
An alphabet $\abc$ is a finite set with at least two symbols, and we let $\abc^* := \bigcup_{i\ge0} \abc^i$ (known as the Kleene closure of $\abc$), and let $\varepsilon$ denote the empty word.
We also define $\abc^{\leq n} := \cup_{i=1}^n \abc^i$.
If $w\in \abc^*$, then $|w|$ is its length and $w(t)$ its $t$-th symbol.
We let $\varsigma {:} w$ be the word obtained by prepending symbol $\varsigma$ to word $w$.
The complement of a decision problem $X \subseteq \{0, 1\}^*$ is $\overline X := \{0, 1\}^* \setminus X$.
Given a matrix $X\in\real^{n\times m}$, its entry in row $i$ and column $j$ is $X_{i,j}$.
We write $X_{i,*}$ for row $i$ and $X_{*,j}$ for column $j$.
The all-ones column vector of dimension $n$ is denoted $\one_n$.

\subsection{Hidden Markov models}\label{sec:hmm_defs}

We consider finite discrete hidden Markov models over a finite observation alphabet $\abc = \{1, \dots, m\}$.
Whenever $m$ is clear from the context, we leave it implicit in the notation.
In fact, we later show that a binary alphabet ($m = 2$) is enough for our hardness results; see Proposition~\ref{thm:binalphabet}.

Formal definitions about HMMs are introduced next, without particular intuition, and we refer to \citet{rabiner1989tutorial} for a tutorial.

\begin{definition}
A \emph{hidden Markov model} (HMM) over $\abc$ is a tuple
$\model=(n,\pi,A,B)$ where $n\in\nat$ is the number of hidden states,
$\pi\in\real^n$ is the initial distribution,
$A\in\real^{n\times n}$ is the transition matrix, and
$B\in\real^{n\times m}$ is the output matrix.
They satisfy
\[
  \pi_i\ge0,\quad \sum_i\pi_i=1,\qquad
  A_{ij}\ge0,\quad \sum_j A_{ij}=1,\qquad
  B_{iy}\ge0,\quad \sum_y B_{iy}=1.
\]
We write $\hmmsn n$ for the set of $n$-state HMMs over the current alphabet and
$\hmms:=\bigcup_{n\ge1}\hmmsn n$.
\end{definition}

\begin{definition}
    An HMM $\model = (n, \pi, A, B)$ induces the following stochastic process:
    \begin{align*}
    \prob_\model(\rand q_1 = i) &:= \pi_i,
    \\ \prob_\model(\rand q_{t+1} = j \mid \rand q_{t} = i) &:= A_{ij},
    \\ \prob_\model(\rand z_{t} = k \mid \rand q_{t} = i) &:= B_{ik}.
    \end{align*}
\end{definition}

It follows that $\prob(\rand q_{t+1} = j) = \sum_{i=1}^n \prob(\rand q_{t} = i) \, A_{ij}$ and $\prob(\rand z_t = k) = \sum_{i=1}^n \prob(\rand q_t = i) \, B_{ik}$.

We use the emit-then-transition convention: the hidden state emits an observation, then transitions to the next state.

\begin{definition}[\cite{rabiner1989tutorial}]\label{def:hmm_output_probability}
For $\model=(n,\pi,A,B)$ and $w=w_1\cdots w_\ell\in\abc^\ell$,
\[
\prob_\model(w)
=
\sum_{q_1,\ldots,q_\ell\in \{1, \dots, n\}}
\pi_{q_1}B_{q_1,w_1}
\prod_{t=1}^{\ell-1} A_{q_t,q_{t+1}}B_{q_{t+1},w_{t+1}} .
\]
This formula sums over $n^\ell$ hidden-state sequences--an exponential computation.
Remark~\ref{rem:forward_backward} gives the usual polynomial-time forward computation of the same quantity (see \citet{rabiner1989tutorial}).
\end{definition}

We next define output equivalence for HMMs:

\begin{definition}\label{def:equivalence}
HMMs $\model$ and $\model'$ over the same alphabet are \emph{output equivalent}, written $\model\equiv\model'$, if and only if
$\prob_\model(w)=\prob_{\model'}(w)$ for every $w\in\abc^*$.
\end{definition}

\subsection{HMM parameterization}

An HMM is specified by the entries of $\pi$, $A$, and $B$, usually expressed in relation to a parameter vector $\theta$. 
A \emph{parameterization} specifies how these entries depend on $\theta$ and which parameter values are allowed. 
Parameterizations are standard in identifiability~\citep{catchpole1997detecting,cole2019parameter}.

To take a parameterized HMM as input to a decision problem, we require its description to be finite, so that it can be encoded on, say, the tape of a Turing machine~\citep{arora2009computational}.
As \citet{sullivant2023algebraic} argues, ``\emph{statistical models are semialgebraic sets}''.
Adopting this view, we consider semialgebraic parameterizations of HMMs, which indeed admit finite descriptions and capture the overwhelming majority of models studied in the literature.

\begin{definition}[HMM parameterization]\label{def:parameterization}
An \emph{HMM parameterization} is a pair $(\Theta,\phi)$, where
$\Theta\subseteq\real^K$ is semialgebraic (see Section~\ref{sec:algebraic_geometry}), and $\phi:\Theta\to\hmmsn n$ is given by polynomials with rational coefficients.
Thus every entry of $\pi$, $A$, and $B$ is a rational polynomial in $\theta_1,\ldots,\theta_K$.
We call $\phi$ the \emph{parameterization}, $\theta\in\Theta$ a \emph{parameter value}, and $\phi(\Theta) := \{\phi(\theta) \mid \theta \in \Theta\}$ the corresponding \emph{family}.
\end{definition}

The domain $\Theta$ being semialgebraic (see Section~\ref{sec:algebraic_geometry}), we encode the defining formula and the polynomials explicitly, with rational coefficients and monomial exponents in binary. 
The input size includes these descriptions and the lists of parameters and HMM entries.

For instance, HMMs are most commonly provided by explicitly listing the entries of $\pi$, $A$, and $B$, which can be viewed as a trivial parameterization where each parameter bijectively maps to each entry.
The stochasticity constraints for $(\pi,A,B)$ may simply be included in the description of $\Theta$.

\subsection{Probabilistic automata}

We provide the explicit definition of probabilistic automata, as these appear in Section~\ref{sec:simple_hard}.

\begin{definition}[Probabilistic automaton; \citet{paz1971introduction}]
\label{def:pa}
A \emph{probabilistic automaton} (PA) is a tuple
$\mathcal A=(Q,\Sigma,\alpha,\{M_a\}_{a\in\Sigma},\eta)$,
where $Q$ is a finite state set of size $n$ (with a fixed ordering), $\Sigma$ is a finite alphabet,
$\alpha\in[0,1]^n$ is a stochastic row vector, every
$M_a\in[0,1]^{n\times n}$ is row-stochastic, and
$\eta\in[0,1]^n$ is a column vector of final probabilities.
For a word $w=a_1\cdots a_\ell$, write
$M_w:=M_{a_1}\cdots M_{a_\ell}$ and $M_\varepsilon:=I$.
The acceptance probability of $w$ is $L_{\mathcal A}(w):=\alpha M_w\eta$.
Two PAs are \emph{equivalent} if and only if they have the same acceptance probability for every word.
\end{definition}

\subsection{The existential theory of the reals}\label{sec:etr}

The existential theory of the reals (ETR; see \citet{schaefer2026existential} for a survey), is the set of all true formulas of the form ``$\exists x_1, \dots, x_k \in \real$, $P(x_1, \dots, x_k) = 0$'', where $P$ is a quantifier-free system of polynomials with rational coefficients.
The decision problem for ETR is, given such a formula, to determine whether the formula holds or not.
The complexity class $\etr$ is defined as the downward closure of the existential theory of the reals under polynomial-time many-one (Karp) reductions.
The complement of $\etr$ is denoted $\coetr$.

Classical ETR decision algorithms place
$\etr \subseteq \pspace$~\citep{canny1988some,renegar1992computational}, and
$\np \subseteq \etr$ by encoding Boolean satisfiability with variables constrained by $x (1-x)=0$, negation by $\neg x = 1 - x$, and conjunction by $x y = 1$~(see for instance \cite{hillar2013most}).
The decision problem of determining if a quadratic ETR formula holds is denoted $\quadetr$, and it is easy to show that $\quadetr$ is $\etr$-complete (see \cite{schaefer2026existential}), since any polynomial can be transformed into a (polynomially-sized) system of quadratic equations by introducing helper variables (for instance, $x^3 = 0$ becomes $x^2 - y = 0, xy = 0$).

Allowing inequalities, disequalities, and Boolean combinations does not change the class.
For instance, $x \ge 0$ can be written in purely equational form as $\exists y \in \real, x - y^2 = 0$, disequality $x \ne 0$ can be written as $\exists y \in \real, xy - 1 = 0$, and $x > 0$ can be written as $\exists y \in \real, x y^2 - 1 = 0$.
Finite Boolean disjunction can be encoded by selector variables; for instance, $x - a = 0 \, \vee \, y - b = 0$ becomes $(x - a)z = 0, (y - b)(1-z) = 0, z(1-z)=0$.
Logical negation is then treated as an abbreviation for negating the comparison relation (i.e.\ $\neg (x = 0)$ is just $x \neq 0$), and implication is an abbreviation for disjunction (i.e.\ $x = 0 \to y = 0$ becomes $x \neq 0 \vee y = 0$).

We will repeatedly use \emph{$\etr$ certificates}.
To prove that a problem is in $\etr$, we describe a polynomial-size existential real formula whose satisfying assignments are exactly the desired witnesses.
We use such formulas for non-identifiability: a satisfying assignment consists of a rival HMM together with auxiliary variables certifying output equivalence.

\subsection{Basic definitions from algebraic geometry}
\label{sec:algebraic_geometry}

Here we provide basic definitions and refer to \citet{hartshorne2013algebraic} for details.
A set $S\subseteq\real^d$ is \emph{algebraic} if it is the common zero set of finitely many real polynomials with rational coefficients. 
It is \emph{semialgebraic} if membership in $S$ can be expressed by a finite Boolean combination of equations and inequalities of rational polynomials. 
Equivalently, a semialgebraic set is a finite union of sets, each defined by finitely many simultaneous polynomial equations and inequalities. 
In inputs to our computational decision problems, all coefficients are rational.

For a nonempty semialgebraic set $S\subseteq\real^d$, its \emph{real dimension}, denoted $\dim S$, is the largest $r$ such that projecting $S$ onto some $r$ of its coordinates gives a set with nonempty interior in $\real^r$. 
Thus a nonempty finite set has dimension $0$, a line or parabola has dimension $1$, a nonempty open subset of $\real^d$ has dimension $d$, and so on. 
In Section~\ref{sec:generic_identifiability}, we use the property that existentially quantifying some variables in a formula of polynomial equations and inequalities still defines a semialgebraic set in the remaining variables~\citep{basu2006algorithms}.

\section{Polynomial-size certificates of HMM output equivalence}\label{sec:hmm_equivalence_certificate}

This section constructs the output-equivalence certificate used in our later complexity upper bounds.
For HMMs with $n$ and $n'$ hidden states, the classical theorem by \citet{paz1971introduction} says that, if the models are different in their output probabilities, then some word of length at most $n+n'-1$ witnesses the difference.
Thus equivalence could be checked by comparing all words up to this length, but that list is exponential in the number of hidden states.

The so-called forward procedure gives a more compact encoding.
After reading a word $w$, the probability that an HMM is in a given hidden state is given by the row vector $\pi\transpose T_w$. 
Although there are exponentially many words, all these vectors live in an $n$-dimensional vector space, so at most $n$ such can be linearly independent.
Tzeng's algorithm exploits this to decide probabilistic automaton equivalence in polynomial time by keeping only a basis of the vectors generated by words, and \citet{kiefer2014stability} show how the condition that this basis remains closed after reading another symbol can be expressed in $\etr$.
We extend the same idea to HMMs and apply it in parallel to the target and rival HMM.
The resulting certificate replaces exponentially many word equalities by polynomially many real variables and equations.

\subsection{$\etr$-certificates for HMM output equivalence}

First, we recall a well-established idea from classical literature known as the forward procedure, which we later use in our $\etr$ certificate constructions. 
The forward procedure evaluates a word by a linear pass of matrix multiplications instead of explicitly summing over all hidden-state paths like done in Definition~\ref{def:hmm_output_probability}.
We refer the reader to Section~III.A in \citet{rabiner1989tutorial} for an instructive tutorial.

\begin{remark}[Forward procedure]\label{rem:forward_backward}
For $y\in\abc$, let $D(y):=\diag(B_{1y},\ldots,B_{ny})$, $T_y:=D(y)A$.
For $w=w_1\cdots w_\ell$, let $T_w:=T_{w_1}\cdots T_{w_\ell}$ and
$T_\varepsilon:=I_n$.
Then, $\prob_\model(w)=\pi\transpose T_w\one_n$.
Note that the final multiplication by $A$ is harmless because $A\one_n=\one_n$.
\end{remark}

In the theorem below we provide an $\etr$ certificate $(F, \{\vv*{T}{y}\}_{y\in\abc})$ for HMM output equivalence.
The rows of $F$ provide the basis of the \emph{forward space}, $V := \linspan \{(\pi\transpose,-{\pi'}\transpose)T_w^+ \mid w \in \abc^*\}$, 
which is a technique we borrow from the probabilistic automaton equivalence algorithm by~\citet{tzeng1992polynomial} and is also used in the $\etr$ certificate of probabilistic automaton non-minimality of \citet{kiefer2014stability}.
The equation $FT_y^+=\vv*{T}{y}F$ says that, after appending one more observation $y$, every such vector is still a linear combination of the rows of $F$.
In Section 3.1 of \citet{kiefer2014stability}, a similar equation is defined for probabilistic automaton non-minimality, which they present as a Brzozowski-like minimization (as in \citet{brzozowski1962canonical}).

\begin{theorem}[$\exists\mathbb{R}$ certificate for HMM output equivalence]
\label{thm:kw-certificate}
Let $\model=(n,\pi,A,B)$ and $\model'=(n',\pi',A',B')$ be HMMs over the same alphabet.
Then, $\model\equiv\model'$ if and only if there exist
$F\in\real^{(n+n')\times(n+n')}$ and
$\vv*{T}{y}\in\real^{(n+n')\times(n+n')}$ for $y\in\abc$ such that
\begin{align}
F_{1,*}&=(\pi\transpose,-{\pi'}\transpose), \label{eq:equiv-cert-firstrow}\\
F\one_{n+n'}&=0, \label{eq:equiv-cert-orthogonal}\\
FT_y^+&=\vv*{T}{y}F
\qquad\text{for every }y\in\abc, \label{eq:equiv-cert-invariance}
\end{align}
where we let $T_y$ and $T'_y$ be the matrices from Remark~\ref{rem:forward_backward} to define
\[
T_y^+:=
\begin{pmatrix}
T_y&0\\
0&T'_y
\end{pmatrix}
\in\real^{(n+n')\times(n+n')}.
\]
\end{theorem}

\begin{proof}
For $w=w_1\cdots w_\ell$, we write $T_w^+:=T_{w_1}^+\cdots T_{w_\ell}^+$ and $T_\varepsilon^+:=I_{n+n'}$.
Then,
\begin{equation}\label{eq:cert-to-prob}
(\pi\transpose,-{\pi'}\transpose)T_w^+\one_{n+n'}
=
\prob_\model(w)-\prob_{\model'}(w).
\end{equation}
Therefore, $\model \equiv \model'$ if and only if $(\pi\transpose,-{\pi'}\transpose)T_w^+\one_{n+n'} = 0$ for all $w \in \abc^*$.

$(\Longleftarrow)$ direction: if matrices $F$ and $\{\vv*{T}{y}\}_{y\in\abc}$ satisfying \eqref{eq:equiv-cert-firstrow}--\eqref{eq:equiv-cert-invariance} exist, then $\model \equiv \model'$.

Suppose that such $F$ and $\{\vv*{T}{y}\}_{y\in\abc}$ exist.
We claim that, for every word $w\in\abc^*$, the row vector $(\pi\transpose,-{\pi'}\transpose)T^{+}_w$ belongs to the row span of $F$ (i.e., is a linear combination of the rows in $F$), which we prove by induction on $w$.
This is immediate for $w=\varepsilon$ by~\eqref{eq:equiv-cert-firstrow}.
In the inductive case, if $(\pi\transpose,-{\pi'}\transpose)T_w^+=uF$, then for all $y\in\abc$,
\[
(\pi\transpose,-{\pi'}\transpose)T_{wy}^+
=
uFT_y^+
=
u\vv*{T}{y}F,
\]
so $(\pi\transpose,-{\pi'}\transpose)T_{wy}^+$ is still in the row span of $F$.
By~\eqref{eq:equiv-cert-orthogonal}, every vector $v$ in the row span of $F$ satisfies $v \one_{n+n'} = 0$, and equation~\eqref{eq:cert-to-prob} thus gives $\prob_\model(w)=\prob_{\model'}(w)$ for every word $w$.

$(\Longrightarrow)$ direction: if $\model\equiv\model'$ then matrices $F$ and $\{\vv*{T}{y}\}_{y\in\abc}$ satisfying \eqref{eq:equiv-cert-firstrow}--\eqref{eq:equiv-cert-invariance} exist.

As in \citet{tzeng1992polynomial}, consider the forward space
\[
V:=
\linspan
\left\{
(\pi\transpose,-{\pi'}\transpose)T_w^+
\mid w\in\abc^*
\right\}.
\]
Since $V$ has dimension $d \leq n + n'$, choose a basis $b_1,\ldots,b_d$ of $V$ such that $b_1=(\pi\transpose,-{\pi'}\transpose)$.
Then, construct matrix $F$ by letting the first $d$ rows be this basis and padding the remaining with zeros.
By construction of $F$, its row-span equals to $V$.
Equation \eqref{eq:equiv-cert-firstrow} thus holds immediately.
Moreover, because $\model \equiv \model'$ and by \eqref{eq:cert-to-prob}, we have that $(\pi\transpose,-{\pi'}\transpose)T_w^+\one_{n+n'} = 0$ for every word $w \in \abc^*$.
Hence every vector in $V$ is orthogonal to $\one_{n+n'}$, and we conclude $F \one_{n+n'} = 0$, i.e.~\eqref{eq:equiv-cert-orthogonal} is satisfied.

It remains to construct matrices $\vv*{T}{y}$ satisfying \eqref{eq:equiv-cert-invariance}.
For each $y\in\abc$, the space $V$ is invariant under right multiplication by $T^{+}_y$; that is, every vector $(\pi\transpose, -{\pi'}\transpose)T^{+}_w$ of $V$ satisfies
\[
  ((\pi\transpose,-{\pi'}\transpose)T^{+}_w)T^{+}_y = (\pi\transpose,-{\pi'}\transpose)T^{+}_{wy}\in V.
\]
Therefore, each row of $FT^{+}_y$ lies in the row-span of $F$.
For every row of $FT^{+}_y$, choose coefficients expressing that row as a linear combination of the rows of $F$. 
Then, placing these coefficient rows into a matrix gives $\vv*{T}{y}\in\real^{(n+n')\times (n+n')}$ with $FT^{+}_y=\vv*{T}{y}F$.
Thus equations \eqref{eq:equiv-cert-firstrow}--\eqref{eq:equiv-cert-invariance} are satisfied by $F$ and $\{\vv*{T}{y}\}_{y \in \abc}$.
\end{proof}

\begin{remark}
The equations above are polynomial in the HMM entries and in the entries of the certificate matrices.
There are $O(m(n+n')^2)$ auxiliary variables and polynomially many equations.
Thus Theorem~\ref{thm:kw-certificate} gives a polynomial-size $\etr$ certificate for HMM output equivalence, exactly in the form needed below when the rival HMM is itself existentially quantified.
\end{remark}

\section{Deterministic identifiability up to equality}
\label{sec:deterministic_hmm_equality}

The target model in our decision problems is always fully specified. 

\begin{problem}\label{problem:iddeteq}
Let $\iddeteq$ denote the following problem.
The input is an HMM parameterization $\phi:\Theta\to\hmmsn n$, and a rational parameter value $\theta^\star\in\Theta$ defining the target $\target:=\phi(\theta^\star)$.
The decision problem is whether, for every $\theta'\in\Theta$, if $\target\equiv\phi(\theta')$ then $\target=\phi(\theta')$.
\end{problem}

We write $\coiddeteq:=\overline{\iddeteq}$ for the corresponding non-identifiability problem.
Thus, $(\phi,\theta^\star)\in\coiddeteq$ iff there exists $\theta'\in\Theta$ such that
$\phi(\theta')\equiv\target$ but $\phi(\theta')\ne\target$.
We use the same naming convention for the complements of the other identification problems throughout.

Once a target $\phi(\theta^\star)$ is fixed, we call any other HMM $\phi(\theta')$ allowed by the parameterization $\Theta$ a \emph{rival HMM}.
Here ``rival'' only means a candidate alternative to the target and it does not presuppose that the two HMMs are distinct or output equivalent.

\subsection{Deterministic identifiability up to equality is in $\coetr$}

\begin{theorem}
\label{thm:det-eq-etr-certificate}
$\coiddeteq$ is in $\etr$, and therefore $\iddeteq$ is in $\coetr\subseteq\pspace$.
\end{theorem}

\begin{proof}
Existentially guess a rival parameter value $\theta'\in\real^K$ together with the matrices
$F$ and $\{\vv*{T}{y}\}_{y\in\abc}$ from Theorem~\ref{thm:kw-certificate}.
The formula contains:
\begin{enumerate}[label=(\roman*)]
    \item the polynomial constraints defining $\theta'\in\Theta$;
    \item equations~\eqref{eq:equiv-cert-firstrow}--\eqref{eq:equiv-cert-invariance}, with the rival entries given by $\phi(\theta')$;
    \item the non-equality of the two HMMs $\|\phi(\theta')-\phi(\theta^\star)\|_2^2>0$.
\end{enumerate}
All entries of $\phi(\theta')$ are polynomial in $\theta'$, so every constraint above is polynomial.
If the formula is satisfiable, $\phi(\theta')$ is exactly a distinct output-equivalent rival.
Conversely, any such rival supplies $\theta'$, and Theorem~\ref{thm:kw-certificate} supplies the equivalence matrices.
Hence $\coiddeteq\in\etr$.
\end{proof}

\subsection{Deterministic identifiability up to equality is $\coetr$-hard}\label{sec:simple_hard}

Equivalently, we prove the $\etr$-hardness of $\coiddeteq$, which requires encoding arbitrary polynomial equations as a non-identifiability problem.
For fully general parameterizations, as in Definition~\ref{def:parameterization}, $\etr$-hardness follows trivially: the entries are already arbitrary polynomials, so the equations can be encoded directly.
It is therefore more productive to consider the computational hardness of restricted classes of parameterizations.

This section covers \emph{simple} parameterizations in the sense of  \citet{junges2018finite} and \citet{winkler2019complexity}, where every entry is a rational number, a variable, or $1$ minus a variable.
Another kind of restricted class, known as masked or rectangular parameterizations, is covered later in Section~\ref{sec:mask_hard}.
We start with the definition of simple parameterizations:

\begin{definition}[Simple parameterizations; \citet{junges2018finite}]
A parameterization $\phi:\Theta\to\hmmsn n$ is \emph{simple} if every HMM entry is either a rational constant, a parameter $\theta_i$, or its complement $1-\theta_i$, and
\[
    \Theta := \{\theta\in[0,1]^K\mid \phi(\theta)\text{ is stochastic}\}.
\]
\end{definition}

Our approach is to first establish the $\coetr$-hardness of the identifiability of probabilistic automata (PA; Definition~\ref{def:pa}) with simple parameterizations, and then reduce to HMM identifiability with simple parameterizations. 
We use the same meaning of \emph{simple} for PAs: every entry of
$\alpha$, the matrices $M_a$, and $\eta$ is either a rational constant, a parameter $\theta_i$, or its complement $1-\theta_i$. 

The following example explains how an $\etr$ question becomes a question about whether a simply parameterized PA has a distinct equivalent PA. 
The approach is to extend the given system of polynomial equations by adding a dummy solution, and then encode it as a PA with simple parameters where the dummy solution is the target PA, and any distinct rival equivalent to the target corresponds precisely to a solution to the original system of equations.

\begin{example}[$\exists\mathbb{R}$ to PA identifiability]
\label{ex:pa_identifiability}
Consider the system
\begin{equation}\label{eq:example_pa_hard}
xy=\tfrac14,\qquad x+y=1,\qquad x,y\in[0,1].
\end{equation}
To add a dummy solution, introduce a selector variable $z\in[0,1]$ and extend~\eqref{eq:example_pa_hard} as follows
\begin{equation}\label{eq:example_pa_hard_extended}
\begin{gathered}
 zxy = z/4, \quad z(x+y) = z, \quad z(1-z)=0, \quad (1-z)x = 0, \quad (1-z)y = 0.
\end{gathered}
\end{equation}
The equation $z(1-z)=0$ forces $z\in\{0,1\}$. If $z=0$, the last two equations force $x=y=0$, giving the dummy solution. If $z=1$, the remaining constraints reduce to the original system~\eqref{eq:example_pa_hard}.

Next, we encode system~\eqref{eq:example_pa_hard_extended} as a PA with simple parameterizations:
initially split into a target branch and a rival branch (Figure~\ref{fig:example_pa_hard}~\subref{fig:gadget_a}), chain transitions to express products (Figure~\ref{fig:example_pa_hard}~\subref{fig:gadget_b}), and merge branches to compute weighted sums~(Figure~\ref{fig:example_pa_hard}~\subref{fig:gadget_c}).
Find the full diagram in Appendix~\ref{app:example_pa_hard}.
Its only words that can have nonzero acceptance probability are
\begin{align*}
 L_{\mathcal A}(\sigma_1\sigma_p\sigma_1\sigma_1)&=zxy+(1-z)/4, &
 L_{\mathcal A}(\sigma_1\sigma_s\sigma_1)&=z(x+y)/2+(1-z)/2,\\
 L_{\mathcal A}(\sigma_1\sigma_x\sigma_1)&=(1-z)x, &
 L_{\mathcal A}(\sigma_1\sigma_y\sigma_1)&=(1-z)y,\\
 L_{\mathcal A}(\sigma_1\sigma_z)&=z(1-z).
\end{align*}
By taking $x=y=z=0$ as the target PA, 
there exists a distinct equivalent rival (within this family of PAs) if and only if~\eqref{eq:example_pa_hard} has a solution. 
Here, that rival is $x=y=1/2, z = 1$.
\end{example}

\tikzset{
  gadget/.style={
    ->, >=latex, auto, node distance=3cm,
    transform shape,      %
    every state/.style={minimum size=0.9cm, inner sep=1pt},
    every node/.append style={font=\small},  %
  }
}
\newcommand{\gadgetscale}{0.8}
\begin{figure}
    \centering
    \begin{subfigure}[b]{0.2\linewidth}
        \centering
        \begin{tikzpicture}[gadget,scale=\gadgetscale]
            \node[state] at (0, 0)  (q0) {$q_\text{init}$};
            \node[state] at (2.5, 1)  (z)  {$q_z$};
            \node[state] at (2.5, -1) (nz) {$q_{\bar z}$};

            \draw[->] (0, 1) -- node[right] {$1$} (q0);
            \draw
            (q0) edge node[above,sloped] {$\sigma_1, z$}   (z)
            (q0) edge node[below,sloped] {$\sigma_1, 1-z$} (nz)
            ;
        \end{tikzpicture}
        \caption{Initial gadget.}\label{fig:gadget_a}
    \end{subfigure}
    \hfill
    \begin{subfigure}[b]{0.4\linewidth}
        \centering
        \begin{tikzpicture}[gadget,scale=\gadgetscale]
            \node[state] at (0, 0)    (z)   {$q_z$};
            \node[state] at (0, -2)   (nz)  {$q_{\bar z}$};
            \node[state] at (2, 0)  (xy)  {$q_{xy}$};
            \node[state] at (4, 0)    (y)   {$q_y$};
            \node[state] at (2, -2) (c2)  {$q_{c_2}$};
            \node[state] at (4, -2)   (c1)  {$q_{c_1}$};
            \node[state,accepting] at (6, -1) (acc) {$q_\text{acc}$};

            \draw
            (z)  edge node[above] {$\sigma_p, 1$}   (xy)
            (xy) edge node[above] {$\sigma_1, x$}   (y)
            (y)  edge node[above,sloped] {$\sigma_p, y$}   (acc)
            (nz) edge node[above] {$\sigma_p, 1/4$} (c2)
            (c2) edge node[above] {$\sigma_1, 1$}   (c1)
            (c1) edge node[below,sloped] {$\sigma_1, 1$}   (acc)
            ;
        \end{tikzpicture}
        \caption{Product gadget.}\label{fig:gadget_b}
    \end{subfigure}
    \hfill
    \begin{subfigure}[b]{0.35\linewidth}
        \centering
        \begin{tikzpicture}[gadget,scale=\gadgetscale]
            \node[state] at (0, 0)   (z)  {$q_z$};
            \node[state] at (0, -2)  (nz) {$q_{\bar z}$};
            \node[state] at (2.5, 0)  (x)  {$q_{x}$};
            \node[state] at (2.5, -1)   (y)  {$q_y$};
            \node[state] at (2.5, -2)  (c1) {$q_{c_1}$};
            \node[state,accepting] at (5, -1) (acc) {$q_\text{acc}$};

            \draw
            (z)  edge node[above] {$\sigma_s, 1/2$} (x)
            (z)  edge node[below,sloped] {$\sigma_s, 1/2$} (y)
            (x)  edge[bend left] node[above] {$\sigma_1, x$}   (acc)
            (y)  edge node[above] {$\sigma_1, y$}   (acc)
            (nz) edge node[above] {$\sigma_s, 1/2$} (c1)
            (c1) edge[bend right] node[above] {$\sigma_1, 1$}   (acc)
            ;
        \end{tikzpicture}
        \caption{Sum gadget.}\label{fig:gadget_c}
    \end{subfigure}

    \caption{Gadgets in a simply parameterized PA encoding \eqref{eq:example_pa_hard_extended}. Transitions are labeled with \emph{letter, probability} pairs, and the alphabet is $\{\sigma_1, \sigma_s, \sigma_p, \sigma_x, \sigma_y, \sigma_z\}$.}
    \label{fig:example_pa_hard}
\end{figure}
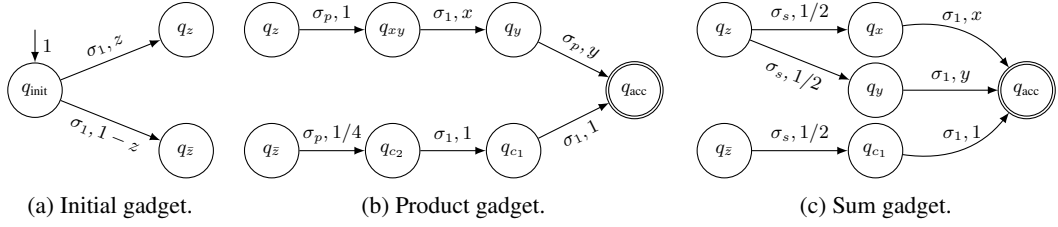

Example~\ref{ex:pa_identifiability} extends to general systems of polynomial equations. 
Encoding polynomial equations in PA acceptance probabilities already appears in \citet[Theorem~5.6]{lenhardt2009probabilistic}, \citet[Remark~7]{chonev2017reachability}, and \citet[Lemma~20]{winkler2019complexity}.
Our proof of Theorem~\ref{thm:pa_identifiability_hard} takes inspiration from this prior work to prove the $\etr$-hardness of PA identifiability, and later in Theorem~\ref{thm:iddeteq_simple_etrhard} the result is extended to HMMs with simple parameters.

To construct a PA for a general system of polynomial equations, we move negative terms to the opposite side and divide both sides by a sufficiently large constant.
Then, let the PA assign a branch to each term, where the branch probability supplies the coefficient, and transitions multiply the variables. 
All accepting paths merge after reading the same word, so their probabilities sum to the value of the polynomial.
The challenge is that if both sides of the equation depend on (unknown) rival parameters, their values are unknown in advance and cannot be encoded as a target acceptance probability.
We overcome this by using one word per side and multiplying both acceptance probabilities by a shared scalar parameter.%
If the two sides are equal, this scalar parameter can make both probabilities match a given positive target acceptance probability.
Find the full proof in Appendix~\ref{app:pa_dentifiability_hard}.

Formally, the input to \emph{PA non-identifiability up to equality} is a simply parameterized PA $\mathcal A(\theta)$ and a rational parameter value $\theta^\star$ specifying its target. 
The question is whether the same family contains a PA equivalent to $\mathcal A(\theta^\star)$ but different in at least one numerical entry, with state labels fixed.

\begin{restatable}{theorem}{paetrhard}\label{thm:pa_identifiability_hard}
Probabilistic automaton non-identifiability under simple parameterizations is $\etr$-hard.
\end{restatable}

We next turn the PA into an HMM that randomly chooses a word, simulates the PA on that word, and emits a symbol indicating acceptance or rejection. 
Conversions between automata and HMMs are well studied; see \citet[Proposition~9]{casacuberta1990some} and \citet[Section~3]{dupont2005links}. 
These transformations combine parameters such that the parameterizations are no longer simple. 
Moreover, our PA assigns an acceptance probability to each supplied word, whereas the HMM must generate its observations. 
We therefore choose the input letters with fixed probabilities and record both the chosen word and the PA's acceptance decision in the observations.

For our reduction, %
each PA entry is copied directly into an HMM entry, so the resulting parameterization remains simple and distinct PAs give distinct HMMs. 
The probability of observing a word followed by acceptance is the PA's acceptance probability multiplied by a known positive factor. 
This correspondence also preserves equivalence, as verified in Appendix~\ref{app:hmm_identifiability_simple}, and yields the following theorem.

\begin{restatable}{theorem}{hmmetrhard}\label{thm:iddeteq_simple_etrhard}
$\coiddeteq$ is $\etr$-hard (and thus complete) even under simple parameterizations.
Equivalently, $\iddeteq$ is $\coetr$-hard under simple parameterizations.
\end{restatable}

Appendix~\ref{app:pa-to-hmm-simple} provides an example of this PA-to-HMM transformation.

\subsection{Deterministic identifiability up to equality with masks is $\conp$-hard}\label{sec:mask_hard}

This section proves the $\np$-hardness of \emph{masked} parameterizations, sometimes deemed \emph{rectangular}~\citep{wiesemann2013robust,jansen2022parameter}, which further restrict simple parameterizations by requiring every open entry to be an independent parameter.
In contrast, simple parameterizations as in Section~\ref{sec:simple_hard} allow the same parameter to occur in several HMM entries.
Another way to view masks is that, given a target HMM $\target$, it specifies which entries of the rival have to match the target exactly, and which are ``fully open'', in that they are allowed to take any value.

\begin{definition}[Mask parameterizations]
For a rational target $\target=(n,\pi^\star,A^\star,B^\star)$, a \emph{mask} $\mu=(\mu_\pi,\mu_A,\mu_B)$ assigns $0$ or $1$ to every entry of $(\pi,A,B)$.
Entries with mask value $1$ are fixed to the corresponding target value, while every entry with mask value $0$ receives its own distinct parameter.
The parameter domain consists exactly of the values for which the resulting entries are stochastic.
We denote the induced family by $\masked(\target,\mu)$.
Thus masks are rectangular simple parameterizations: apart from the row-sum constraints required for stochasticity, distinct open entries are independent.
\end{definition}

To prove hardness under masks, we encode a Boolean assignment in the emission probabilities of an HMM. For each variable, two states represent its positive and negative literals. Requiring certain observation probabilities to match those of the target forces these states to represent opposite truth values. Further observation probabilities enforce the clauses of a Boolean formula. Different clauses refer to the same literal states, so the construction does not need to assign the same parameter to different open entries.

Starting from a $\mathsf{3SAT}$ formula, we first add a Boolean variable to obtain a new formula with a known satisfying assignment. Every other satisfying assignment of this new formula encodes a solution of the original formula. We use the known assignment to specify the target HMM. The construction ensures that each satisfying assignment of the new formula determines exactly one equivalent HMM, including all open transition and emission probabilities. Consequently, the target has a distinct equivalent rival exactly when the original formula is satisfiable. Appendix~\ref{app:mask_nphard} gives the full proof.

\begin{restatable}{theorem}{masknphard}\label{thm:iddeteq_mask_nphard}
$\coiddeteq$ is $\np$-hard even under mask parameterizations.
Equivalently, $\iddeteq$ is $\conp$-hard under mask parameterizations.
\end{restatable}

\section{Deterministic identifiability up to state renaming and bisimilarity}

A target HMM may be unidentifiable for the sole reason that a trivial output-equivalent rival exists, for instance by renaming the hidden states of the target HMM to obtain the rival.
This phenomenon is well documented in the literature, and thus it is common to discard rivals that are ``trivially equivalent'' to the target~\citep{leroux1992maximum,tune2013hidden}.
We study two notions of ``trivially equivalent'' here: 1) equivalence up to hidden state renaming in Section~\ref{sec:permutation}, as often done in the identifiability literature; and 2) bisimilarity in Section~\ref{sec:bisimilarity}, a fundamental concept in theoretical computer science.

\subsection{Deterministic identifiability up to state renaming is in $\forall \etr$}
\label{sec:permutation}

We next broaden our definition of identifiability to align with the practice of discarding rivals obtained by state-renaming.
A common way to remove state-renaming ambiguity in statistical applications is to impose an ordering convention on some state-specific parameter~\citep{leroux1992maximum,cole2019parameter}.
We instead treat permutation equivalence explicitly in the decision problem.

\begin{definition}[Permutation equivalence]
Let $\model=(n,\pi,A,B)$ and $\model'=(n,\pi',A',B')$.
They are \emph{permutation equivalent}, written $\model\similarperm\model'$, if there exists a permutation matrix $\Pi\in\{0,1\}^{n\times n}$ (i.e.\ exactly one 1 in each row and column, and all other entries set to 0) such that 
\[
\pi'=\Pi\pi, \qquad A'=\Pi A\Pi\transpose, \qquad B'=\Pi B.
\]
\end{definition}

\begin{problem}
Let $\iddetperm$ denote the following problem.
The input is an HMM parameterization $\phi:\Theta\to\hmmsn n$, and a rational parameter value $\theta^\star\in\Theta$ defining the target $\target := \phi(\theta^\star)$. 
The decision problem is whether, for every $\theta'\in\Theta$, if $\target \equiv\phi(\theta')$ then $\target \similarperm\phi(\theta')$.
\end{problem}

Even when both HMMs are fully specified, checking whether one is obtained from the other by state permutation is not known to have a polynomial-time algorithm.
The reason is that it reduces from the classical graph isomorphism problem ($\gi$), which asks whether two finite graphs become identical after permuting their vertices~\citep{arora2009computational}.
See Appendix~\ref{app:permutation_GI_complete} for a proof.

\begin{restatable}{lemma}{thmpermequivhard}\label{thm:perm_equiv_hard}
Deciding whether two rational HMMs are permutation equivalent is $\gi$-complete.
\end{restatable}

Thus we do not assume that permutation equivalence itself can be checked in polynomial time.
For the formula below we encode ``the rival is different from every permutation of the target'' as an explicit statement, which introduces an additional universal quantifier.

\begin{proposition}
$\coiddetperm$ is in $\exists\forall\real$. Thus $\iddetperm$ is in $\forall\exists\real\subseteq\pspace$.
\end{proposition}

\begin{proof}
    Existentially guess $\theta'\in\Theta$ and the output-equivalence certificate from Theorem~\ref{thm:kw-certificate}.
    Then universally quantify a real matrix $\Pi\in\real^{n\times n}$ and impose
    \begin{enumerate}[label={}]
        \item[(i)--(iii)] as in Theorem~\ref{thm:det-eq-etr-certificate}.
        \item[(iv)] Either $\Pi$ is not a valid permutation matrix, 
        \begin{align*}
        \Big(\bigvee_{ij} (\Pi_{ij}(1-\Pi_{ij}) \neq 0)\Big) \vee 
        \Big(\bigvee_{j}\sum_{i}\Pi_{ij} \neq 1\Big) \vee
        \Big(\bigvee_{i}\sum_{j}\Pi_{ij} \neq 1\Big),
        \end{align*}
        or the permuted rival is different from the target,
        \begin{align*}
        \Big( \sum_i(\pi'_i-(\pi\Pi\transpose)_i)^2+
        \sum_{i,j}(A'_{ij}-(\Pi A\Pi\transpose)_{ij})^2+
        \sum_{i,y}(B'_{iy}-(\Pi B)_{iy})^2>0 \Big).
        \end{align*}
    \end{enumerate}
    To introduce constraint (iv) we use disjunction ($\vee$) encoded in $\etr$ as explained in Section~\ref{sec:etr}.
\end{proof}

\subsection{Deterministic identifiability up to bisimilarity is in $\coetr$}
\label{sec:bisimilarity}

With a motivation similar to that in Section~\ref{sec:permutation}, we discard rivals that are bisimilar to the target.
The term ``bisimulation'' is used for two different notions in the literature relevant here.
\citet{doyen2008equivalence} define a relation on state distributions that coincides with output equivalence.
If we used that relation as our notion of uniqueness, every output-equivalent rival would automatically count as equivalent to the target, so the identifiability problem would be trivial.
We therefore use the classical state-based probabilistic bisimulation: hidden states may be grouped together only when they have the same observation distribution and send the same total transition probability to every group.

\begin{definition}[Classical HMM bisimulation, \citet{larsen1989bisimulation}]
Let $\model=(n,\pi,A,B)$ and $\model'=(n',\pi',A',B')$, and consider the disjoint union of their states, $Q := \{1, \dots, n\} \uplus \{1, \dots, n'\}$.
An equivalence relation $\mathcal R \subseteq Q \times Q$ is a \emph{bisimulation} iff, $i\,\mathcal R\,j$ implies
\begin{enumerate}[label=(\roman*)]
    \item the observation distributions agree, i.e. $B_{iy}=B'_{jy}$ for $y = 1, \dots, m$.
    \item for every equivalence class $C$ of $\mathcal R$, the transition probability from $i$ to $C$ equal the transition probability from $j$ to $C$.
\end{enumerate}
The two HMMs are \emph{bisimilar}, written $\model\sim\model'$, if their initial distributions assign the same total mass to every class of the coarsest such relation, i.e. the relation whose equivalence classes are as large as possible while satisfying the two conditions above.
This is the usual state-based probabilistic bisimulation specialized to HMMs~\citep{larsen1989bisimulation,hansson1994time}.
\end{definition}

\begin{problem}
Let $\iddetbisim$ denote the following problem.
The input is an HMM parameterization $\phi:\Theta\to\hmmsn n$, and a rational parameter value $\theta^\star\in\Theta$ defining the target $\target := \phi(\theta^\star)$. 
The decision problem is whether, for every $\theta'\in\Theta$, if $\target \equiv\phi(\theta')$ then $\target \sim \phi(\theta')$.
\end{problem}

For two explicit finite HMMs, classical bisimilarity can be checked by a standard partition-refinement algorithm.
Start by grouping states that have the same observation distribution.
Then repeatedly split a group whenever two states in it send different total transition probability to some current group.
When no further split is possible, the resulting groups are exactly the classical bisimulation classes~\citep{larsen1989bisimulation,hansson1994time}.
The proof below encodes these polynomially many refinement steps directly in an $\etr$ formula. 
Find the full certificate in Appendix~\ref{app:bisim-etr}.

\begin{restatable}{proposition}{thmbisimetr}\label{thm:bisim-etr}
$\coiddetbisim$ is in $\etr$, and therefore $\iddetbisim$ is in $\coetr$.
\end{restatable}

\subsection{Hardness of deterministic identifiability up to state renaming and bisimilarity}
\label{sec:perm_bisim_hard}

The reduction for simple parameterizations from Theorem~\ref{thm:pa_identifiability_hard} can also enforce that the rival be disequal from the target up to bisimilarity. 
We add two new branches, one for $z=0$ and another for $z=1$ using fresh output letters, where both branches are always output-equivalent yet they are non-bisimilar.
This extension survives the PA-to-HMM conversion, as proved in Appendix~\ref{app:bisim_simple_hardness}.

\begin{restatable}{corollary}{corbisimetrhard}\label{cor:iddetbisim_simple_etrhard}
$\coiddetbisim$ is $\etr$-hard (and thus complete) even under simple parameterizations.
Equivalently, $\iddetbisim$ is $\coetr$-hard under simple parameterizations.
\end{restatable}

\begin{remark}\label{rem:hardness_carry_over}
A permutation of hidden states always gives a bisimilar HMM, as $\similarperm \subsetneq \sim$. 
Thus the non-bisimilar rival constructed for Corollary~\ref{cor:iddetbisim_simple_etrhard} cannot be a permutation of the target. 
When the input polynomial system is unsatisfiable, the same reduction leaves the target as the only observationally equivalent HMM in the family. These two properties establish the following corollary using the same reduction.

\end{remark}

\begin{corollary}\label{cor:iddetperm_simple_etrhard}
$\coiddetperm$ is $\etr$-hard (and thus complete) even under simple parameterizations.
Equivalently, $\iddetperm$ is $\coetr$-hard under simple parameterizations.
\end{corollary}

For masks, the reduction of Theorem~\ref{thm:iddeteq_mask_nphard} already ensures that every distinct observationally equivalent rival is non-bisimilar to the target. 
Such a rival changes the Boolean variable $z$ from false to true. 
The reachable state representing the positive literal $z$ then has an emission distribution that no target state has: its identifying observation distinguishes it from the other literal states, and the comparison states have additional observations in their emission distributions. 
Bisimulation cannot match this reachable state to any target state. 
Appendix~\ref{app:bisim_mask_hardness} gives the details.

\begin{restatable}{corollary}{corbisimnphard}\label{cor:iddetbisim_mask_nphard}
$\coiddetbisim$ is $\np$-hard even under mask parameterizations.
Equivalently, $\iddetbisim$ is $\conp$-hard under mask parameterizations.
\end{restatable}

The same reduction also proves hardness up to state permutation, similarly to Remark~\ref{rem:hardness_carry_over}. 

\begin{corollary}\label{cor:iddetperm_mask_nphard}
$\coiddetperm$ is $\np$-hard even under mask parameterizations.
Equivalently, $\iddetperm$ is $\conp$-hard under mask parameterizations.
\end{corollary}

\section{Binary-alphabet, local, generic, and family-wise identifiability}\label{sec:variants}

We consider various identifiability subproblems in this section.
Section~\ref{sec:binary-alphabet} restricts the deterministic identifiability problem to binary alphabets and shows that computational hardness is preserved.
Section~\ref{sec:local_identifiability} proves that local identifiability is in $\coetr$, Section~\ref{sec:generic_identifiability} proves that generic identifiability is in $\exists \coetr \cap \forall \etr$, and Section~\ref{sec:family-wise} proves that family-wise identifiability is in $\coetr$.

\subsection{Binary-alphabet subproblems}\label{sec:binary-alphabet}

The HMMs constructed for proving $\etr$- and $\np$-hardness in Theorems~\ref{thm:iddeteq_simple_etrhard} and~\ref{thm:iddeteq_mask_nphard} have alphabets that grow in size linearly with respect to the input.
Therefore, the computational hardness of the deterministic identification of HMMs with a fixed alphabet size remains open.
In Proposition~\ref{thm:binalphabet}, we transform HMMs with arbitrary alphabets into binary alphabets ($\abc = \{0,1\}$) in polynomial time such that simple and mask parameterizations are preserved, meaning that the hardness theorems carry over.
We start by defining the binary alphabet identifiability problem.

\begin{problem}
Let $\biniddeteq$ denote the restriction of $\iddeteq$ to a binary alphabet.
\end{problem}

Problem $\biniddeteq$ is already an instance of $\iddeteq$, so the nontrivial direction is to reduce in the converse direction.
The reduction is done via a routine binary encoding of the alphabet symbols, the main detail to take care of being that simple and mask parameterizations remain so after the transformation.
The full proof is given in Appendix~\ref{app:alphabet_independent}.

\begin{restatable}{proposition}{binalphabet}\label{thm:binalphabet}
$\biniddeteq$ and $\iddeteq$ are polynomial-time many-one equivalent.
Moreover, the reduction to a binary alphabet preserves simple and mask parameterizations.
\end{restatable}

As a direct consequence of Proposition~\ref{thm:binalphabet}, the computational hardness properties of $\iddeteq$ carry over to $\biniddeteq$, and the same holds for deterministic identifiability up to state-renaming and bisimilarity when restricted to binary alphabets, given that non-permutation-equivalent and non-bisimilar HMMs remain so after transformation.

\begin{corollary}
    Even when restricted to binary alphabets, $\iddeteq$, $\iddetperm$, and $\iddetbisim$ are $\coetr$-hard with simple parameterizations, and $\conp$-hard with mask parameterizations.
\end{corollary}

\subsection{Deterministic local identifiability}
\label{sec:local_identifiability}

Local identifiability asks whether an output equivalent rival exists within a radius around the target HMM, no matter how small said radius is.
We consider local identifiability with a fixed target HMM, as done so far in this paper, even though local identifiability literature often considers it family-wise rather than target-wise~\citep{catchpole1997detecting,cole2019parameter}.
We cover family-wise identifiability later in Section~\ref{sec:family-wise}.

\begin{problem}\label{problem:local-ident}
Let $\idlocal$ denote the following problem.
Given rational target HMM $\target\in\hmmsn n$, decide whether $\exists \delta>0$ such that, $\forall \rival\in\hmmsn n$, if
$\|\target-\rival\|_2\le\delta$ and
$\target\equiv\rival$, then $\target=\rival$.
\end{problem}

The local non-identifiability problem $\coidlocal$ says that for every radius there is a distinct equivalent rival inside that radius, i.e.,
$\forall \delta > 0$, $\exists \rival \in \hmmsn{n}$,
$0<\|\target-\rival\|_2^2\le \delta^2$ and 
$\target\equiv\rival$.
Therefore, as a trivial consequence of Theorem~\ref{thm:det-eq-etr-certificate}, $\coidlocal$ is in $\forall\exists\real$, and $\idlocal$ is in $\exists\forall\real$, which is contained in $\pspace$ by \citet{canny1988some} and \citet{renegar1992computational}.

We can tighten the bound down to $\etr$ by eliminating a quantifier.
The next proposition uses the bound of
\citet[Theorem~1]{jeronimo2013minimum} to replace the quantification over radius $\delta$ by one fixed (yet extremely small) radius whose description is nevertheless polynomial-size.
The full proof is found in Appendix~\ref{app:local_etr}.

\begin{restatable}{proposition}{thmsmallradius}\label{thm:small-radius-local-id}
For every rational target $\target\in\hmmsn n$, one can compute in polynomial time an integer $L_\target$, polynomially bounded in the input size, such that, if $\target$ is locally identifiable, every distinct equivalent rival satisfies
$\|\target-\rival\|_2^2>2^{-2^{L_\target}}$.
\end{restatable}

We use Proposition~\ref{thm:small-radius-local-id} to show that $\idlocal$ is in $\coetr$ by taking care that the doubly exponentially small radius is encoded in a polynomially-sized $\etr$ formula.

\begin{theorem}
$\coidlocal$ is in $\etr$, and therefore $\idlocal$ is in $\coetr$.
\end{theorem}

\begin{proof}
Let $r_\target:=2^{-2^{L_\target}}$, where $L_\target$ is given by Proposition~\ref{thm:small-radius-local-id}.
We claim that $\target$ is locally non-identifiable iff there exists
$\rival\in\hmmsn n$ such that
\begin{equation}\label{eq:local_claim}
    0<\|\target-\rival\|_2^2\le r_\target
    \qquad\text{and}\qquad
    \target\equiv\rival.
\end{equation}
If~\eqref{eq:local_claim} holds, the contrapositive of Proposition~\ref{thm:small-radius-local-id} says that $\target$ cannot be locally identifiable.
Conversely, if $\target$ is locally non-identifiable, distinct equivalent rivals occur inside every positive radius, hence one satisfies~\eqref{eq:local_claim}.

It remains to encode $r_\target$ without writing its enormous denominator.
Existentially introduce variables $z_0,\ldots,z_{L_\target} \in \real$, and impose constraints
$2z_0=1$ and $z_{\ell+1}=z_\ell^2$ for $\ell=0,\ldots,L_\target-1$.
Then $z_{L_\target}=2^{-2^{L_\target}}=r_\target$.
Together with the stochasticity constraints for $\rival$, the inequality $0<\|\target-\rival\|_2^2\le z_{L_\target}$, and the equivalence certificate of Theorem~\ref{thm:kw-certificate}, this is a polynomial-size existential real formula. Hence $\coidlocal\in\etr$.
\end{proof}

\subsection{Generic identifiability}
\label{sec:generic_identifiability}

The deterministic identifiability studied so far concerns one target HMM. 
Generic identifiability instead asks whether the set of parameter values specifying non-identifiable HMMs form a lower-dimensional subset of the parameter domain $\Theta$, using the parameterizations $(\Theta,\phi)$ from Definition~\ref{def:parameterization},

For generic identifiability up to equality, the parameter values specifying non-identifiable HMMs form the set
\begin{equation}
 \mathsf{Bad}_{\phi}
 :=\bigl\{\theta\in\Theta \mid \exists\theta'\in\Theta,
       \ \phi(\theta')\equiv\phi(\theta)
       \text{ and } \phi(\theta')\ne\phi(\theta)\bigr\}.
 \label{eq:generic-bad-set}
\end{equation}
Observe that two parameter values that give exactly the same entries of $(\pi, A, B)$ do not count toward non-identifiability.
A stronger notion of identifiability can be provided through notions of equivalence different from strict equality as discussed in Remark~\ref{rem:generic_other_equality}.
We next introduce the decision problem.

\begin{problem}[Generic identifiability]\label{prob:generic-identifiability}
Let $\idgeneric$ denote the following problem.
The input is an HMM parameterization $(\Theta, \phi)$.
The decision problem is whether $\dim\mathsf{Bad}_{\phi} < \dim\Theta$.
\end{problem}

\citet[Theorem~11]{dorfler2024complexity} provide a $\pspace$ bound for the generic identifiability of linear structural causal models (SCM).
Even though SCMs are fundamentally different from HMMs both in structure and semantics, we take inspiration from their theorem to provide a $\pspace$ bound on the generic identifiability problem for HMMs.
Our output-equivalence certificate from Theorem~\ref{thm:kw-certificate} gives a polynomial-size existential description of $\mathsf{Bad}_{\phi}$.
The dimension algorithm by \citet{koiran1999real} then expresses each question $\dim S \ge r$ by an existential real formula, for $S = \Theta$ and $S = \mathsf{Bad}_{\phi}$. 
Following \citet[Theorem~11]{dorfler2024complexity}, we combine the candidate values for $r$, the dimension threshold, into one formula. 
For each threshold, the assertion that $\Theta$ reaches that dimension is existential, while the assertion that $\mathsf{Bad}_{\phi}$ does not is universal. 
Since these two quantifiers are independent, we get the bound of $\exists \coetr \cap \forall \etr$.
Appendix~\ref{app:generic-upper-bound} provides the full formula and proof.

\begin{theorem}\label{thm:generic-identifiability-upper}
\label{thm:generic-identifiability-pspace}
$\idgeneric$ is in $\exists \coetr \cap \forall \etr \subseteq \pspace$.
\end{theorem}

\begin{remark}\label{rem:generic_other_equality}
    The bound in Theorem~\ref{thm:generic-identifiability-upper} holds for generic identifiability modulo another equivalence relation $\othereq$, provided the condition $\phi(\theta')\not\othereq\phi(\theta)$ has a quantifier-free real description constructible in polynomial time, because its negation would still then be quantifier free. 
    If $\othereq$ requires quantification, then the bound may shift upward in the real hierarchy~\citep{schaefer2026existential}.
\end{remark}

\subsection{Family-wise identifiability}\label{sec:family-wise}

Our deterministic identifiability problems fix one target HMM and ask whether it has a distinct output-equivalent rival. 
We now ask whether this uniqueness property holds for every HMM in a given family, which is a prevalent question in the parametric identifiability literature~\citep{catchpole1997detecting,cole2019parameter}. 
We next introduce deterministic family-wise identifiability up to equality as a decision problem.

\begin{problem}[Family-wise identifiability]
Let $\iddetfam$ denote the following problem.
The input is an HMM parameterization $(\Theta, \phi)$.
The decision problem is whether, for every $\theta, \theta'\in\Theta$, if $\phi(\theta) \equiv\phi(\theta')$ then $\phi(\theta) = \phi(\theta')$.
\end{problem}

The bound from Theorem~\ref{thm:det-eq-etr-certificate} directly carries over to family-wise identifiability: existentially guess both $\theta$ and $\theta'$, and use constraints (i)--(iii) in Theorem~\ref{thm:det-eq-etr-certificate} to enforce (non-)identifiability.

\begin{lemma}\label{thm:familywise-coetr}
$\iddetfam$ is in $\coetr$.
\end{lemma}

As in Remark~\ref{rem:generic_other_equality}, one can also require uniqueness up to other equivalence relations. 
One may also distinguish between local and global identifiability.
We do not explicitly consider all such variants due to combinatorial explosion, but expect Lemma~\ref{thm:familywise-coetr} to carry over as long as the variant does not introduce additional complexity.

\begin{remark}\label{rem:exhaustive_summary}
    \citet{catchpole1997detecting} and \citet{cole2019parameter} study local identifiability of parameterized families using an \emph{exhaustive summary}, $\kappa(\theta)$: a finite vector of observable quantities that determines the observation distribution.
    Differentiating this summary with respect to the parameters gives the derivative matrix $D(\theta)=\partial\kappa(\theta)/\partial\theta$.
    If the generic rank of $D$ is lower than the number of parameters, the model is parameter redundant and thus non-identifiable, whereas a full generic rank means that the model is at least locally identifiable.
    \citet{cole2019parameter} implements this method using a symbolic algebra system, but it requires exponential space if used as a general solution to the identifiability decision problem here.
\end{remark}

\section{Conclusions and future work}

\paragraph{Conclusions.} We give polynomial-size $\etr$ certificates for HMM output equivalence in Theorem~\ref{thm:kw-certificate} and use them throughout the paper.
Using them, we show deterministic identifiability up to equality is in $\coetr$. 
It is $\coetr$-complete for simple parameterizations and $\conp$-hard for mask (rectangular) parameterizations.
These lower bounds also hold for identifiability up to hidden-state renaming and for bisimilarity.
Identifiability up to hidden-state permutation lies in $\forall \etr$, while identifiability up to bisimilarity lies in $\coetr$.
The hardness results do not depend on the alphabet: identifying HMMs over a binary alphabet is as hard as over an arbitrary one.
We further show that local identifiability is in $\coetr$, generic identifiability is in $\exists\coetr \cap \forall \etr$, and that family-wise identifiability is in $\coetr$.
By \citet{canny1988some} and \citet{renegar1992computational}, all of these problems are decidable in $\pspace$.

\paragraph{Open questions and future work.} Deterministic global identifiability up to equality with mask parameterizations is shown to belong to $\coetr$ and be $\conp$ hard, thus $\coetr$-completeness is open (but unexpected).
The hardness of local and generic identifiability are left open.
Moreover, we do not explicitly handle some combinations of the variants introduced in Section~\ref{sec:variants}, but we believe most results carry over without particular difficulty.
Our results may motivate the development of efficient identifiability checking algorithms in $\pspace$ that may combine SAT solvers for the Boolean constraints like in Proposition~\ref{thm:bisim-etr} with $\etr$ solvers for the polynomial equations in Theorem~\ref{thm:det-eq-etr-certificate}.
These solvers could be used for the exhaustive classification of well-established models based on their identifiability, in a sense extending the work by \citet{cole2019parameter}.

\begin{ack}
    Partially funded by the ERC Starting Grant 101077178 (DEUCE).
\end{ack}

\bibliographystyle{apalike}
\bibliography{references}

\vfill
\pagebreak

\appendix

\section{Deterministic HMM identifiability up to equality with simple parameters}\label{app:pa_dentifiability_hard}

\newcommand{\proofstep}[2]{\vspace{0pt}~\\\emph{#1.\ #2}}

\subsection{Probabilistic automaton identifiability is $\coetr$-hard}\label{app:pa_dentifiability_hard}

\paetrhard*

\begin{proof}
We extend the construction of Example~\ref{ex:pa_identifiability} to general polynomial equations.
The selector variable again separates the dummy solution from solutions of the input equation. 
The main additional task is to compare two polynomial expressions whose common value is not fixed in advance.
We organize the proof into steps I.\ through V.\ for readability.

Our use of separate input words to enforce equations follows the arithmetic encoding in \citet{lenhardt2009probabilistic}, 
\citet{chonev2017reachability}, and \citet{winkler2019complexity}.
We use these ideas to construct a family with a specified rational target, requiring every transition probability to be a rational constant, a parameter, or its complement. 

\proofstep{I}{Preparing the polynomial equation}

We reduce from feasibility of a single equation
\begin{equation}
 P(x)=Q(x),\qquad x\in[0,1]^k,
 \label{eq:pa-source}
\end{equation}
where both polynomials are given explicitly as sums of terms, have degree at most four, and have nonnegative rational coefficients. 
\citet[Lemma~29]{winkler2019complexity} prove the $\etr$-hardness of finding a zero in $[0,1]^k$ of a polynomial $f$ of degree at most four that is nonnegative throughout this domain. 
Its coefficients may be either positive or negative. 
Moving the negative-coefficient terms in $f(x)=0$ to the other side gives an instance of~\eqref{eq:pa-source}, meaning that this problem is $\etr$-hard as well.

Let $C$ be the maximum of $1$, the sum of the coefficients of $P$, and the sum of the coefficients of $Q$. 
Scale down $p=P/C$ and $q=Q/C$, so that each of $p$ and $q$ has coefficient sum at most $1$ and therefore takes values in $[0,1]$ throughout the domain.
Introduce a selector variable $z\in[0,1]$. 
Just as in Example~\ref{ex:pa_identifiability}, the equations
\begin{align}
 z(1-z) &= 0,\label{eq:pa_id_select_z}\\
 (1-z)x_i &= 0 \quad \text{for } i=1, \dots k,\label{eq:pa_id_select_x}\\
 zp(x) &= zq(x)\label{eq:pa_id_select_pq},
\end{align}
allow either the dummy solution $z=0,x_i=0$, or $z=1$ with $x$ solving~\eqref{eq:pa-source}.

In Example~\ref{ex:pa_identifiability}, the system equates positive polynomials to (fixed) constants.
Here in the general case, both $p(x)$ and $q(x)$ may vary with $x$, meaning their common value is unknown.
To compare them without knowing this value, introduce a parameter $\lambda\in[0,1]$. 
We will give two input words the acceptance probabilities
\begin{equation}
 \frac{\lambda}{4}\bigl(1+zp(x)\bigr),\qquad
 \frac{\lambda}{4}\bigl(1+zq(x)\bigr),
 \label{eq:pa-comparison-plan}
\end{equation}
respectively, and require both to match the target value $1/8$.
The constant $1$ in these expressions is important: if $z=1$ and
$p(x)=q(x)=v$, choosing $\lambda=1/(2(1+v))$ makes both probabilities $1/8$. This choice lies in $[1/4,1/2]$ because $v\in[0,1]$. Conversely, matching the positive target value forces $\lambda>0$, so equality of the two probabilities forces $zp(x)=zq(x)$.

\proofstep{II}{Constructing the PA}

We now construct the PA that realizes these equations.
The parameter vector is
\[
 \theta=(z,x_1,\ldots,x_k,\lambda)\in[0,1]^{k+2}.
\]
Use an internal letter $t$ and distinct query letters
\[
 \Sigma=\{t,\sigma_z,\sigma_L,\sigma_R\}\cup\{\sigma_i:1\le i\le k\}.
\]
Each query letter ``tests'' an equation, which gets read first, and then subsequent copies of $t$ carry out that test.
The PA starts in $q_{\mathrm{init}}$ with probability $1$ and has a single accepting state $q_{\mathrm{acc}}$, with final probability $1$. 
Every other state has final probability $0$. 
There is also a rejecting sink $q_{\mathrm{rej}}$.
Different tests use disjoint intermediate states and share only these three states.

We next write $u\xrightarrow[b]{a}v$ for a transition from $u$ to $v$ on letter $b$ with probability $a$.
To test the selector variable by~\eqref{eq:pa_id_select_z}, introduce states $q_{z,0},q_{z,1}$ and transitions
\[
 q_{\mathrm{init}}\xrightarrow[\sigma_z]{1}q_{z,0}
 \xrightarrow[t]{z}q_{z,1}
 \xrightarrow[t]{1-z}q_{\mathrm{acc}}.
\]
To test variable $x_i$ by~\eqref{eq:pa_id_select_x}, introduce states $q_{i,0},q_{i,1}$ and transitions
\[
 q_{\mathrm{init}}\xrightarrow[\sigma_i]{1}q_{i,0}
 \xrightarrow[t]{1-z}q_{i,1}
 \xrightarrow[t]{x_i}q_{\mathrm{acc}}.
\]
At each of these transitions, send the remaining probability to $q_{\mathrm{rej}}$. These paths give acceptance probabilities $z(1-z)$ for $\sigma_z t^2$ and $(1-z)x_i$ for $\sigma_i t^2$.

The tests for equation~\eqref{eq:pa_id_select_pq} are more involved.
We give the construction once for a polynomial $R$: use it with $R=p$ and query letter $\sigma_L$, and then with $R=q$ and query letter $\sigma_R$. The two copies have different states but share $\lambda$.

Because $R$ has degree at most four, we can write it as
\begin{equation}
 R(x)=\sum_{r=1}^{N_R}a_r\prod_{h=1}^{4}v_{r,h},
 \qquad a_r>0,\quad \sum_r a_r\le1,\quad
 v_{r,h}\in\{1,x_1,\ldots,x_k\}.
 \label{eq:pa-terms}
\end{equation}
Here we encode terms (i.e.\ monomials) of degree less than four via factors equal to $1$. 
A constant term has four such factors.
We implement each product by a chain and add the terms by branching, using the polynomial encoding in \citet[Remark~7]{chonev2017reachability} and \citet[Lemma~20]{winkler2019complexity}.
Figure~\ref{fig:example_pa_hard} provides good intuition on how this encoding is achieved.
Unlike their treatment of negative coefficients through complementary factors, we have already placed negative terms on the opposite side of the equation.

For each copy, we do as follows.
Introduce an entry state $q_e$, a branching state $q_b$, five states $q_{d,0},\ldots,q_{d,4}$ for a constant contribution, and five states $q_{r,0},\ldots,q_{r,4}$ for each term $r$ of~\eqref{eq:pa-terms}.
These names are local to the copy being constructed. 
Let $\sigma \in \{\sigma_L, \sigma_R\}$ be its query letter, and set
\[
 q_{\mathrm{init}}\xrightarrow[\sigma]{1}q_e \xrightarrow[t]{\lambda}q_b,
\]
and direct the remaining probability mass, $1-\lambda$, to state $q_\mathrm{rej}$ as usual.
On the next $t$, split the probability at $q_b$ as follows:
\begin{align*}
q_b &\xrightarrow[t]{1/4} q_{d,0},&&\\
q_b &\xrightarrow[t]{a_r/4} q_{r,0} &\text{for } r &= 1, \dots, N_R,\\
q_b &\xrightarrow[t]{w} q_\mathrm{rej} &\text{where } w &= 3/4-\tfrac14\sum_r a_r.
\end{align*}
Thus one branch will supply the constant ($\nicefrac\lambda4$) contribution, and each remaining successful branch corresponds to one term of $R$. 
Observe that the rejecting probability is nonnegative.

To specify the rest of the paths, let
$q_{d,5}=q_{\mathrm{acc}}$ and $q_{r,5}=q_{\mathrm{acc}}$.
The constant branch uses five probability-$1$ transitions:
\[
 q_{d,h} \xrightarrow[t]{1} q_{d,h+1}\quad \text{for } h=0, \dots, 4.
\]

On the branch for term $r$, to test~\eqref{eq:pa_id_select_pq}, first test $z$ and then its four factors:
\begin{align*}
q_{r,0} &\xrightarrow[t]{z} q_{r,1},
&q_{r,0} &\xrightarrow[t]{1-z}q_{\mathrm{rej}},\\
q_{r,h} &\xrightarrow[t]{v_{r,h}}q_{r,h+1},
&q_{r,h} &\xrightarrow[t]{1-v_{r,h}} q_{\mathrm{rej}},\quad \text{for } h= 1, \dots, 4.
\end{align*}
Every successful path from $q_e$ to $q_{\mathrm{acc}}$ reads exactly seven copies of $t$: one for $\lambda$, one for the branching step, and five for the chosen branch. In particular, the constant contribution and all polynomial terms contribute to the acceptance probability of the same word $\sigma t^7$.

Every state-letter row specified above sums to $1$.
For any state-letter transition not yet specified, put probability $1$ on $q_{\mathrm{rej}}$. 
In particular, $q_{\mathrm{rej}}$ loops on every letter, every transition from $q_{\mathrm{acc}}$ goes to $q_{\mathrm{rej}}$, and a query letter read inside a test leads to rejection.

\proofstep{III}{Verifying the acceptance probabilities of the constructed PA}

In the PA just constructed, the constant branch contributes $\lambda/4$ to the probability of $\sigma t^7$. 
The branch for term $r$ contributes
$\lambda a_r z\prod_{h=1}^4v_{r,h}/4$.
Adding these contributions gives $\lambda(1+zR(x))/4$, as required. 
Consequently the complete list of words with potentially nonzero acceptance probability is
\begin{align}
 L_{\mathcal A(\theta)}(\sigma_z t^2)&=z(1-z),&\nonumber\\
 L_{\mathcal A(\theta)}(\sigma_i t^2)&=(1-z)x_i &\quad \text{for } i=1, \dots, k,\nonumber\\
 L_{\mathcal A(\theta)}(\sigma_L t^7)&=\tfrac{\lambda}{4}(1+zp(x)),\nonumber&\\
 L_{\mathcal A(\theta)}(\sigma_R t^7)&=\tfrac{\lambda}{4}(1+zq(x))&.
 \label{eq:pa-language}
\end{align}
To see that the list is exhaustive, an accepting path must start with exactly one query letter and then follow its test to $q_{\mathrm{acc}}$. Stopping earlier leaves the PA in a nonaccepting state, and reading an extra letter after reaching $q_{\mathrm{acc}}$ leads to $q_{\mathrm{rej}}$. Reading any wrong letter also leads to that sink. Thus all unlisted words have probability $0$, independently of the parameters.

\proofstep{IV}{Correctness of the reduction}

Choose the rational target parameter value $\theta^\star := (z^\star, x_i^\star (\text{for } i=1, \dots, k), \lambda^\star)$ as follows:
\[
 z^\star=0,\qquad x_i^\star=0\ \quad (\text{for } i=1,\dots,k),\qquad
 \lambda^\star=1/2.
\]
The target PA gives probability $0$ to the selector and variable words $\sigma_zt^2$ and $\sigma_it^2$, and $1/8$ to both words for comparing $p$ and $q$, that is $\sigma_Lt^7$ and $\sigma_Rt^7$.
This target is known without solving the input equation.

Suppose rival $\phi(\theta')$ is output equivalent to the target.
The word $\sigma_z t^2$ forces $z(1-z)=0$, so $z$ is either $0$ or $1$.
If $z=0$, each word $\sigma_i t^2$ forces $x_i=0$, and either word for comparing $p$ and $q$ forces $\lambda/4=1/8$, hence $\lambda=1/2$.
Therefore, $\phi(\theta') = \phi(\theta^\star)$, thus $\theta'$ can be discarded as identifiability requires that the rival must be distinct from the target.

It follows that a distinct output-equivalent rival must have $z=1$. 
The two comparison words require
\[
 \frac{\lambda}{4}(1+p(x))=\frac18
 =\frac{\lambda}{4}(1+q(x)).
\]
The positive middle value implies $\lambda>0$. Subtracting the two expressions and dividing by $\lambda/4$ yields $p(x)=q(x)$, which is equivalent to $P(x)=Q(x)$. 
Therefore every rival that is distinct from the target yet output equivalent supplies a solution of~\eqref{eq:pa-source}.

Conversely, let $x$ solve~\eqref{eq:pa-source}. Set $z=1$ and choose
\[
 \lambda=\frac{1}{2(1+p(x))}.
\]
Notice this is a valid parameter value in $[1/4,1/2]$.
The selector and variable words, $\sigma_zt^2$ and $\sigma_it^2$, have probability $0$, and both comparison words, $\sigma_Lt^7$ and $\sigma_Rt^7$, have probability $1/8$ because $p(x)=q(x)$. 
All other words have probability $0$ by~\eqref{eq:pa-language} and the argument on the exhaustive list of words above. 
Hence the resulting PA is output equivalent to the target, yet it is distinct as $z^\star=0$ and $z' = 1$.

Finally, we verify that the reduction in polynomial-time.
There are two intermediate states for each selector or variable test, and $2+5(1+N_R)$ intermediate states for each of the two polynomial tests, where $N_R$ is the number of terms of $R\in\{p,q\}$. 
The alphabet has $k+4$ letters. 
Since $P$ and $Q$ are given as explicit sums and have degree at most four, these quantities are polynomial in the input size. 
Sums and quotients of the input rational coefficients, including $C$, the normalized coefficients, and the fixed branch probabilities, have polynomial bit length. 
\end{proof}

\subsection{Example $\coetr$ to PA identifiability}\label{app:example_pa_hard}

Figure~\ref{fig:example_pa_hard_full} provides the full transition function of the PA described in Example~\ref{ex:pa_identifiability}.
The alphabet is $\{\sigma_1,\sigma_p,\sigma_s,\sigma_x,\sigma_y,\sigma_z\}$. 
An edge label is \emph{letter, probability}. 
State $q_{\mathrm{init}}$ is initial and $q_{\mathrm{acc}}$ is the only accepting state. 
For each state and letter, all remaining probability goes to an omitted rejecting sink state, which loops on every letter. 
In particular, every transition from $q_{\mathrm{acc}}$ goes to that sink. All parameters lie in $[0,1]$.

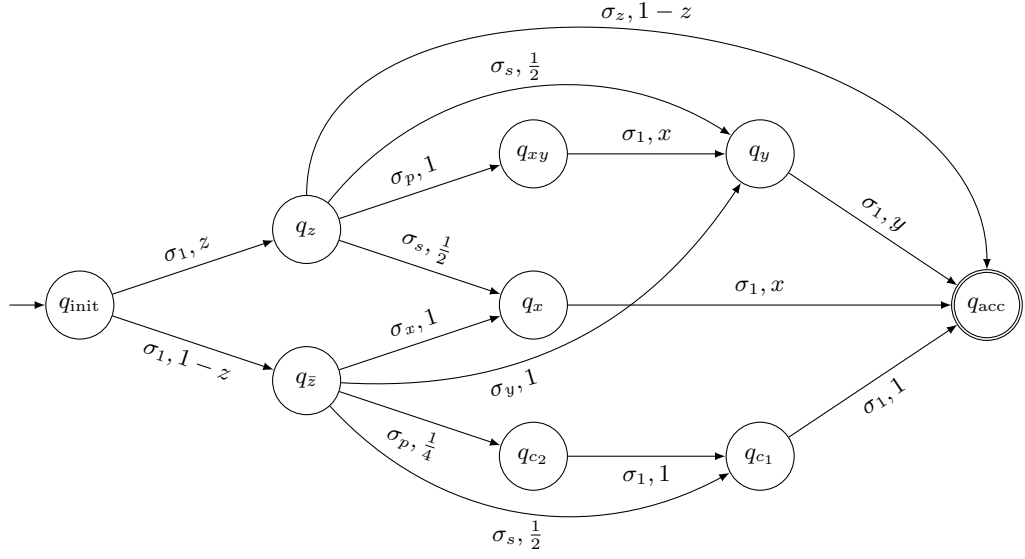
\begin{figure}
\centering
\begin{tikzpicture}[gadget]
\node[state,initial,initial text=] (q) at (0,0) {$q_{\mathrm{init}}$};
\node[state] (Z) at (3,1) {$q_z$};
\node[state] (N) at (3,-1) {$q_{\bar z}$};
\node[state] (P) at (6,2) {$q_{xy}$};
\node[state] (X) at (6,0) {$q_x$};
\node[state] (Y) at (9,2) {$q_y$};
\node[state] (D) at (6,-2) {$q_{c_2}$};
\node[state] (E) at (9,-2) {$q_{c_1}$};
\node[state,accepting] (F) at (12,0) {$q_{\mathrm{acc}}$};
\path (q) edge node[above,sloped] {$\sigma_1,z$} (Z)
(q) edge[->] node[below,sloped] {$\sigma_1,1-z$} (N)
(Z) edge[->] node[above,sloped] {$\sigma_p,1$} (P)
(P) edge[->] node[above] {$\sigma_1,x$} (Y)
(Y) edge[->] node[above,sloped] {$\sigma_1,y$} (F)
(Z) edge[->] node[above,sloped] {$\sigma_s,\tfrac12$} (X)
(Z) edge[->, bend left=42] node[above] {$\sigma_s,\tfrac12$} (Y)
(X) edge[->] node[above,sloped] {$\sigma_1,x$} (F)
(N) edge[->] node[above,sloped] {$\sigma_x,1$} (X)
(N) edge[->, bend right=30] node[below,sloped,pos=.35] {$\sigma_y,1$} (Y)
(N) edge[->] node[below,sloped] {$\sigma_p,\tfrac14$} (D)
(D) edge[->] node[below] {$\sigma_1,1$} (E)
(N) edge[->, bend right=38] node[below] {$\sigma_s,\tfrac12$} (E)
(E) edge[->] node[below,sloped] {$\sigma_1,1$} (F);
\draw[->] (Z) to[out=90,in=90,looseness=1.0] node[above] {$\sigma_z,1-z$} (F);
\end{tikzpicture}
\caption{A simply parameterized PA encoding polynomial system $xy=1/4; x+y=1$. The target parameters $x=y=z=0$ are given, and finding a distinct output-equivalent rival exactly corresponds to solving the system. In this case, such rival exists: $x=y=\tfrac12$ (and $z=1$).}
\label{fig:example_pa_hard_full}
\end{figure}

\section{HMM identifiability with simple parameters is $\coetr$-hard}\label{app:hmm_identifiability_simple}

\hmmetrhard*

\begin{proof}
We reduce PA non-identifiability up to equality to $\coiddeteq$. 
The input is a simply parameterized PA $\mathcal A(\theta^\star) = (Q,\Sigma,\alpha,\{M_a\}_{a\in\Sigma},\eta)$,
where $\theta^\star \in \Theta$ is a rational parameter value specifying the target PA. 
We construct an HMM family $\phi(\Theta)$ with the same parameter domain $\Theta$ and target parameter values $\theta^\star$. 
The construction will preserve both output equivalence and equality, so a distinct output-equivalent rival exists in one family exactly if and only if it exists in the other.

Correspondences between PAs and HMMs are well studied in the literature, where one adds states to replace transitions that emit observations by states that emit observations---see for instance \citet[Proposition~9]{casacuberta1990some} and \citet[Lemma~4]{dupont2005links}. 
The latter paper's Lemmas~3 through 5 and Proposition~8 establish polynomial-size transformations between PAs and HMMs. 
We apply this general idea in this proof; however, the existing reductions are insufficient as-is since they do not preserve simple parameterizations.
In particular, Lemma~3 by \citet{dupont2005links} forms sums and ratios of transition entries, while Proposition~9 by \citet{casacuberta1990some} and Lemma~4 by \citet{dupont2005links} can multiply initial and transition probabilities. 
The transformation provided below carefully ensures that the parameterization remains simple.

\proofstep{I}{Constructing the HMM}

Let $Q=\{q_1,\ldots,q_n\}$, $m=|\Sigma|$, and $\kappa=1/(m+1)$. 
For each PA state $q_i$, introduce a control state $q_{i,\mathrm c}$, one letter state $q_{i,a}$ for each $a\in\Sigma$, and a decision state $q_{i,\mathrm s}$. 
Note these three types of states are disjoint, because the subscripts $\mathrm c$ and $\mathrm s$ denote state types, not input letters. 
Add one further state $q_\bot$. 

Recall that we use the emit-then-transition convention (see Section~\ref{sec:hmm_defs}).
The observation alphabet is $\Gamma=\Sigma\cup\{\#,\cmark,\xmark,\bot\}$, where the four added symbols are fresh. 
The symbols $\cmark$ and $\xmark$ record acceptance and rejection, respectively, and $\bot$ is repeated after the decision.
A control state emits $\#$ and then chooses, with equal probability $\kappa$, either one of the $m$ letters or the decision step. 
If it chooses letter $a$, the corresponding letter state emits $a$ and simulates the PA transition $M_a$. 
If it chooses the decision step, the decision state emits $\cmark$ with the PA's final probability and $\xmark$ with the complementary probability, then moves to $q_\bot$.

Formally, let $A$ and $B$ be the HMM transition and emission matrices. 
Its initial distribution is $\pi(q_{i,\mathrm c})=\alpha(q_i)$ and zero on all other states. 
Its only possibly nonzero entries are
\begin{align*}
 A(q_{i,\mathrm c},q_{i,a})&=\kappa,
 &A(q_{i,\mathrm c},q_{i,\mathrm s})&=\kappa,\\
 A(q_{i,a},q_{j,\mathrm c})&=M_a(q_i,q_j),
 &A(q_{i,\mathrm s},q_\bot)&=1,\\
 A(q_\bot,q_\bot)&=1,\\[2pt]
 B(q_{i,\mathrm c},\#)&=1,
 &B(q_{i,a},a)&=1,\\
 B(q_{i,\mathrm s},\cmark)&=\eta(q_i),
 &B(q_{i,\mathrm s},\xmark)&=1-\eta(q_i),\\
 B(q_\bot,\bot)&=1.
\end{align*}
The equations range over all $i,j = 1, \dots, n$ and $a\in\Sigma$, and all unlisted entries are zero. 
Notice that each control row has $m+1$ entries equal to $\kappa$, each letter row copies a stochastic PA row, and all remaining transition and emission probabilities sum to $1$.

\proofstep{II}{Preserving output equivalence}

For a PA word $w=a_1\cdots a_\ell$, define the observation word
\[
 E(w)=\#a_1\#\cdots a_\ell\#,
 \qquad E(\varepsilon)=\#.
\]
This word records the simulated letters, with $\#$ marking each control step. 
Write $\prob_{\phi(\theta)}(u)$ for the probability that the HMM generates finite word $u$.
For every $r \ge 0$, we get
\begin{align}
 \prob_{\phi(\theta)}(E(w)\cmark\bot^r)
 &=\kappa^{\ell+1}L_{\mathcal A(\theta)}(w),\label{eq:hmm-simple-accept}\\
 \prob_{\phi(\theta)}(E(w)\xmark\bot^r)
 &=\kappa^{\ell+1}\bigl(1-L_{\mathcal A(\theta)}(w)\bigr).
 \label{eq:hmm-simple-reject}
\end{align}
Indeed, choosing the $\ell$ specified letters and then the decision step contributes the factor $\kappa^{\ell+1}$. 
Conditional on these choices, the simulated PA has state distribution $\alpha M_w$ when the decision is made. 
Acceptance therefore has probability $\alpha M_w\eta$, and rejection has probability $\alpha M_w(\one_n-\eta)=1-\alpha M_w\eta$. 
Once the decision symbol has been emitted, every subsequent observation is $\bot$.

We must also account for words that stop before the decision symbol. The remaining possible nonempty observation words have the forms $E(w)$ and $E(w)a$, for $a\in\Sigma$, with probabilities
\[
 \prob_{\phi(\theta)}(E(w))=\kappa^\ell,
 \qquad
 \prob_{\phi(\theta)}(E(w)a)=\kappa^{\ell+1}.
\]
These probabilities do not depend on $\theta$, because summing over the simulated PA states gives $\alpha M_w\one_n=1$. 
The empty word has probability $1$, and every word outside these forms has probability $0$.
Thus equivalent PAs give the same probability to every HMM observation word. 
Conversely, equivalent HMMs give the same probability to $E(w)\cmark$ for every $w$, and dividing~\eqref{eq:hmm-simple-accept} by the positive factor $\kappa^{|w|+1}$ gives us equality of all PA acceptance probabilities. 
Hence
\begin{equation}
 \mathcal A(\theta)\equiv\mathcal A(\theta')
 \quad\Longleftrightarrow\quad
 \phi(\theta)\equiv\phi(\theta').
 \label{eq:hmm-simple-equivalence}
\end{equation}

\proofstep{III}{Preserving distinctness and simple parameters}

Every PA entry can be read from a designated HMM entry:
\[
 \alpha(q_i)=\pi(q_{i,\mathrm c}),\qquad
 M_a(q_i,q_j)=A(q_{i,a},q_{j,\mathrm c}),\qquad
 \eta(q_i)=B(q_{i,\mathrm s},\cmark).
\]
Consequently, different numerical PAs give different numerical HMMs. 
Conversely, if two PAs have identical numerical entries, every formula defining their HMMs gives the same value. 
With state labels fixed, we therefore have
\begin{equation}
 \mathcal A(\theta)=\mathcal A(\theta')
 \quad\Longleftrightarrow\quad
 \phi(\theta)=\phi(\theta').
 \label{eq:hmm-simple-equality}
\end{equation}
This establishes distinctness of models (rather than only distinctness of parameter vectors).

The HMM family $\phi(\Theta)$ uses exactly the PA parameters, and its entries are copied PA entries, complements of final probabilities, or fixed rational constants. 
The complement of a simple entry is again simple, with $1-(1-\theta_i)=\theta_i$. 
Their product appears in path probabilities but never in a single HMM entry. 
The state set is fixed for the whole family, including parameter values at which some transitions vanish.
Therefore, $(\phi, \Theta)$ is a simple parameterization.

\proofstep{IV}{Correctness of the reduction.}

The HMM has $n(m+2)+1$ states and $m+4$ observation symbols. 
Its only new nontrivial rational constant is $\kappa$, which has $O(\log(m+1))$ bits. 
All matrices and the rational target $\phi(\theta^\star)$ can therefore be written down in polynomial time. 
No additional free entries are introduced.
By~\eqref{eq:hmm-simple-equivalence} and~\eqref{eq:hmm-simple-equality}, the target PA has a distinct equivalent member of its family if and only if the target HMM does. 
Therefore, Theorem~\ref{thm:pa_identifiability_hard} proves $\etr$-hardness of $\coiddeteq$ under simple parameterizations. 
Together with Theorem~\ref{thm:det-eq-etr-certificate}, this gives completeness. 
\end{proof}

\subsection{Example of the PA-to-HMM conversion}\label{app:pa-to-hmm-simple}

\begin{example}\label{ex:pa-to-hmm-simple}
Consider the PA in Figure~\ref{fig:pa-to-hmm-simple}~(a), with parameter $x\in[0,1]$, initial state $q_0$, and sole accepting state $q_0$. 
Its alphabet is $\{a\}$. 
On each $a$, it stays at $q_0$ with probability $x$ and otherwise moves to the absorbing rejecting state $q_1$. 
Thus
\[
 L_{\mathcal A(x)}(a^\ell)=x^\ell\quad(\ell\ge0).
\]

Figure~\ref{fig:pa-to-hmm-simple}~(b) shows the seven-state HMM produced by the transformation. 
Each control state chooses between simulating another $a$ and making the decision, with probability $1/2$ for each choice. 
The letter states copy the PA transition probabilities $x,1-x,1$ directly. 
Because the PA final probabilities are $1$ at $q_0$ and $0$ at $q_1$, all emissions in this example are deterministic. 
In particular,
\[
 \Pr_{\mathcal H(x)}(\#a\#\cmark)=x/4,
 \qquad
 \Pr_{\mathcal H(x)}(\#a\#\xmark)=(1-x)/4.
\]
The factor $1/4$ comes from choosing one letter and then the decision. 
More generally, the observation word $(\#a)^\ell\#\cmark$ has probability $2^{-(\ell+1)}x^\ell$. 
\end{example}

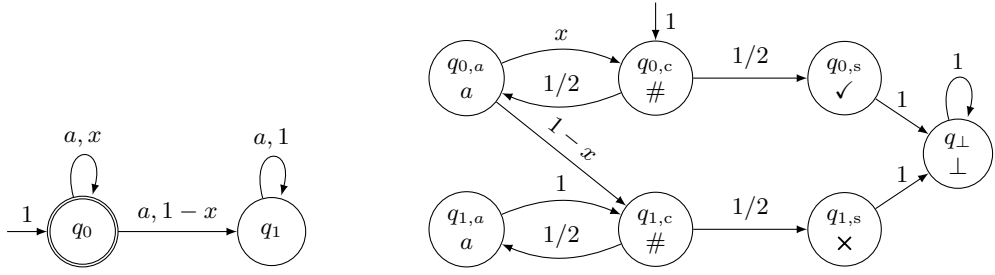
\begin{figure}
    \centering
    \begin{subfigure}[b]{0.3\linewidth}
        \centering
        \begin{tikzpicture}[gadget]
         \node[state,accepting] (q0) at (0,0) {$q_0$};
         \node[state] (q1) at (2.5,0) {$q_1$};
         \path (q0) edge[loop above] node {$a,x$} (q0)
         (q0) edge node[above] {$a,1-x$} (q1)
         (q1) edge[loop above] node {$a,1$} (q1);
         \draw[->] (-1, 0) -- node[above] {$1$} (q0);
        \end{tikzpicture}
        \caption{PA: edge labels are input letter and transition probability.}
        \label{fig:ex1-pa-to-hmm}
    \end{subfigure}
    \hfill
    \begin{subfigure}[b]{0.65\linewidth}
        \centering
        \begin{tikzpicture}[gadget]
             \node[state] (c0) at (0, 0) {$\begin{gathered}q_{0,\mathrm c}\\[-2pt]\#\end{gathered}$};
             \node[state] (a0) at (-2.5, 0) {$\begin{gathered}q_{0,a}\\[-2pt]a\end{gathered}$};
             \node[state] (c1) at (0, -2) {$\begin{gathered}q_{1,\mathrm c}\\[-2pt]\#\end{gathered}$};
             \node[state] (a1) at (-2.5, -2) {$\begin{gathered}q_{1,a}\\[-2pt]a\end{gathered}$};
             \node[state] (s0) at (2.5, 0) {$\begin{gathered}q_{0,\mathrm s}\\[-2pt]\cmark\end{gathered}$};
             \node[state] (s1) at (2.5, -2) {$\begin{gathered}q_{1,\mathrm s}\\[-2pt]\xmark\end{gathered}$};
             \node[state] (d) at (4, -1) {$\begin{gathered}q_\bot\\[-2pt]\bot\end{gathered}$};
             
             \draw[->] (0, 1) -- node[right] {$1$} (c0);
             \path (c0) edge[bend left=23] node[above] {$1/2$} (a0)
             (a0) edge[bend left=23] node[above] {$x$} (c0)
             (a0) edge node[above,sloped] {$1-x$} (c1)
             (c1) edge[bend left=23] node[above] {$1/2$} (a1)
             (a1) edge[bend left=23] node[above] {$1$} (c1)
             (c0) edge node[above] {$1/2$} (s0)
             (c1) edge node[above] {$1/2$} (s1)
             (s0) edge node[above] {$1$} (d)
             (s1) edge node[above] {$1$} (d)
             (d) edge[loop above] node {$1$} (d);
        \end{tikzpicture}
        \caption{HMM: edge labels are transition probabilities, and each state displays its name above the symbol it emits with probability $1$.}
        \label{fig:ex2-pa-to-hmm}
    \end{subfigure}
    \caption{The same parameter $x$ occurs in the PA and its HMM. The HMM first emits the symbol displayed in its current state and then takes a transition. The initial distributions give probability $1$ to $q_0$ and $q_{0,\mathrm c}$, respectively.}
    \label{fig:pa-to-hmm-simple}
\end{figure}

\section{HMM identifiability with masks is $\np$-hard}\label{app:mask_nphard}

\masknphard*

\begin{proof}
We reduce directly from $\mathsf{3SAT}$. 
We will construct a target HMM and a mask such that the target has a distinct output-equivalent rival exactly when the input formula is satisfiable. 

The reduction draws on the use of weighted averages to encode Boolean choices in Kiefer and Wachter's proof of NP-hardness of probabilistic-automaton minimization~\citep[Theorem~5, Propositions~6 and~8]{kiefer2014stability}. 
Our reduction must additionally specify a known target HMM and ensure that different rivals encode different Boolean assignments. %

\proofstep{I}{Preparing the formula}

Let $\varphi(x_1,\ldots,x_N)$ be the input $3$-CNF formula. 
Introduce a fresh Boolean ``selector'' variable $z$ and extend the formula as follows:
\begin{equation}
 \psi(z,x_1, \dots, x_N)=
 \bigwedge_{C\in\varphi}(\neg z\vee C)
 \ \wedge\ \bigwedge_{i=1}^N(z\vee\neg x_i).
 \label{eq:mask-dummy}
\end{equation}
The assignment $\sigma^\star$ with $z=x_1=\cdots=x_N=0$ satisfies $\psi$, and notice that this is the only satisfying assignment with $z=0$ due to the last $N$ clauses. 
When $z=1$, notice that the last $N$ clauses are automatically satisfied, and that the remaining clauses require exactly that $x$ satisfy $\varphi$.
Thus $\psi$ has a satisfying assignment other than $\sigma^\star$ if and only if $\varphi$ is satisfiable.
We will next encode $\psi$ as an HMM and provide assignment $\sigma^\star$ as the parameters of the target HMM, meaning that finding an output equivalent rival precisely maps to finding a satisfying assignment of $\varphi$.

Next, we express clause satisfaction via helper variables.
Write the clauses of $\psi$ as $C_1,\ldots,C_J$, and let $k_j$ be the number of literals in $C_j$. 
Also fix $h_j:=\lceil\log_2(k_j+1)\rceil$ and $L_j:=2^{h_j}$.

Introduce Boolean helper variables $w_{j,0},\ldots,w_{j,h_j-1}$ for each clause $j=1,\dots,J$. For a literal $\ell$, let $\indicator_\ell$ be $1$ when that literal is true and $0$ otherwise. We will enforce
\begin{equation}
 \sum_{\ell\in C_j}\indicator_\ell
 +\sum_{r=0}^{h_j-1}2^r w_{j,r}=L_j.
 \label{eq:mask-clause-bits}
\end{equation}
If clause $C_j$ has $s$ true literals, the added variables give the binary representation of $L_j-s$ by~\eqref{eq:mask-clause-bits}, because they represent precisely the integers from $0$ to $L_j-1$. 
Since $k_j < L_j$, equation~\eqref{eq:mask-clause-bits} has no solution when $s=0$, and has exactly one solution when $1\le s\le k_j$. 
Thus equation~\eqref{eq:mask-clause-bits} enforces clause satisfaction without allowing any additional choice once the original assignment is fixed.

Let $V$ contain $z,x_1,\ldots,x_N$ and the helper variables $w_{j,r}$. 
Extend $\sigma^\star$ to $V$ using the unique binary values given by~\eqref{eq:mask-clause-bits}.

\proofstep{II}{Encoding constraint satisfaction as observation ``checks''}

Recall we use the emit-then-transition convention (Section~\ref{sec:hmm_defs}).
We encode that all conditions to satisfy the Boolean formula are met via ``tests'' on the word probabilities of an HMM.
Each test begins in a \emph{test state}, which emits its own identifying observation and then moves to a state whose emission distribution is involved in the check. 
After this second observation, the HMM moves to a common state $q_\#$ that emits $\#$ forever. 
The initial distribution is fixed and uniform on all test states, and zero elsewhere. 
Denote the number of test states by $R$, and the identifying observation of a test state $u$ by $a_u$ (no other state emits $a_u$).

For any second observation $b$,
\begin{equation}
 \Pr(a_u b)=\frac1R\sum_q A(u,q)B(q,b).
 \label{eq:mask-mixture}
\end{equation}
The sum is a weighted average of the successor states' emission probabilities. 
This is the same way that \citet{kiefer2014stability} turn a convex combination into probabilities of two-letter words induced by a PA in their Proposition~6.
In our HMM setting, the transition probabilities provide the weights and the emission probabilities provide the quantities being averaged.

\proofstep{III}{Encoding literals in emission probabilities}

In the following constructions, all symbols introduced under different names are distinct, and all initial, transition, and emission entries are fixed by the mask unless explicitly declared open. 
In particular, every unlisted entry is fixed to zero by the mask.

For each variable $v\in V$, create two states $q_v^+$ and $q_v^-$, representing the literals $v$ and $\neg v$. 
Introduce observations $y_v^+$ and $y_v^-$ to distinguish these literals in the tests below. 
We fix $\epsilon=1/4$, let $b\in\{+,-\}$, and define the emission row of $q_v^b$ be
\begin{equation}
 B(q_v^b, y^b_v) = \epsilon, \qquad
 B(q_v^b, \mathsf{g}) = 1 - 2\epsilon, \qquad
 B(q_v^b, \mathsf{t}) = r_v^b, \qquad
 B(q_v^b, \mathsf{f}) = \bar r_v^b,
 \label{eq:mask-literal-row}
\end{equation}
and let both $r_v^b$ and $\bar r_v^b$ be open via the mask.
The symbols $\mathsf t,\mathsf f,\mathsf g$ are shared by all these states. Row stochasticity, rather than a shared parameter, gives
\begin{equation}
 r_v^b+\bar r_v^b=\epsilon,\qquad 0\le r_v^b\le\epsilon.
 \label{eq:mask-row-sum}
\end{equation}
To obtain the parameters of the target HMM, $\theta^\star$, take $r_v^b=\epsilon$ exactly when literal $v^b$ is true under $\sigma^\star$.

Next, we will force $r_v^b$ to equal either $0$ or $\epsilon$, interpreting $\epsilon$ as true. 
For each $v\in V$, create two test states $u_v^{\mathsf T}$ and $u_v^{\mathsf F}$. 
The first test will require a literal that emits no $\mathsf f$, and the second will require a literal that emits no $\mathsf t$. 
Introduce four \emph{comparison states} $q_{v,b}^{d}$, with $b\in\{+,-\}$ and $d\in\{\mathsf T,\mathsf F\}$. 
We provide their emission probabilities in the following table:
\begin{equation}
\begin{array}{c|cccccc}
 &c_v^{\mathsf T}&c_v^{\mathsf F}&y_v^b&\mathsf t&\mathsf f&\mathsf g\\\hline
 q_{v,b}^{\mathsf T}&\epsilon&0&\epsilon&\epsilon&0&1-3\epsilon\\
 q_{v,b}^{\mathsf F}&0&\epsilon&\epsilon&0&\epsilon&1-3\epsilon
\end{array}
\label{eq:mask-comparison-rows}
\end{equation}
The observation $c_v^d$ is used only by the two comparison states for that test. 
These fixed distributions let the two different truth assignments produce the same observed weighted average\footnote{The use of these weighted averages is inspired by the two representations of the variable point $p(x_i)$ in Figure~5 and proof of Proposition~8 by~\citet{kiefer2014stability}. 
}.

The only open transitions from $u_v^d$ are its four entries to
\[
 q_{v,+}^d,\quad q_{v,-}^d,\quad q_v^+,\quad q_v^-.
\]
In the target, the $\mathsf T$ test sends probability $1/2$ to the true literal state and $1/2$ to the comparison state for the false literal. The $\mathsf F$ test sends probability $1/2$ to the false literal state and $1/2$ to the comparison state for the true literal. All four entries are open, including those with target value zero.

For either truth assignment, these choices give the following conditional distributions of the second observation, given the identifying first observation of the test has been emitted:
\begin{equation}
\begin{array}{c|ccccccc}
 \text{test}&c_v^{\mathsf T}&c_v^{\mathsf F}&y_v^+&y_v^-&\mathsf t&\mathsf f&\mathsf g\\\hline
 \mathsf T&\epsilon/2&0&\epsilon/2&\epsilon/2&\epsilon&0&1-5\epsilon/2\\
 \mathsf F&0&\epsilon/2&\epsilon/2&\epsilon/2&0&\epsilon&1-5\epsilon/2
\end{array}
\label{eq:mask-test-targets}
\end{equation}
Since each test has fixed initial probability $1/R>0$, matching two-letter word probabilities is equivalent to matching these conditional distributions.

Consider a masked rival that is output-equivalent to the target. 
Write $\alpha_+,\alpha_-$ for its transitions from $u_v^{\mathsf T}$ to the two comparison states and $\gamma_+,\gamma_-$ for its transitions to the two literal states. 
The observations $c_v^{\mathsf T},y_v^+,y_v^-$ give
\begin{equation}
 \alpha_+ + \alpha_- = 1/2,\qquad
 \alpha_+ + \gamma_+ = 1/2,\qquad
 \alpha_- + \gamma_- = 1/2.
 \label{eq:mask-marker-constraints}
\end{equation}
Therefore, $\gamma_++\gamma_-=1/2$, and hence matching the target probability (zero) for $\mathsf f$ gives
\begin{equation}
 \gamma_+(\epsilon-r_v^+)+\gamma_-(\epsilon-r_v^-)=0.
 \label{eq:mask-true-test}
\end{equation}
Every term is nonnegative, and at least one $\gamma_b$ is positive, meaning at least one of $r_v^+,r_v^-$ equals $\epsilon$.

For the $\mathsf F$ test, write $\delta_+,\delta_-$ for the transitions to the literal states. The same three identifying observations give $\delta_++\delta_-=1/2$. Matching probability zero for $\mathsf t$ gives
\begin{equation}
 \delta_+r_v^++\delta_-r_v^-=0,
 \label{eq:mask-false-test}
\end{equation}
so at least one of $r_v^+,r_v^-$ equals $0$. 
Therefore,
\begin{equation}
 \{r_v^+,r_v^-\}=\{0,\epsilon\}.
 \label{eq:mask-boolean}
\end{equation}

Define $v$ to be true exactly when $r_v^+=\epsilon$. 
Once this truth value is fixed, all eight open test transitions are fixed as well. 
For example, if $v$ is true, equation~\eqref{eq:mask-true-test} forces $\gamma_-=0$, so $\gamma_+=1/2$, and then~\eqref{eq:mask-marker-constraints} gives $\alpha_+=0,\alpha_-=1/2$. 
Equation~\eqref{eq:mask-false-test} similarly gives $\delta_+=0,\delta_-=1/2$, and the two $\mathsf F$ comparison weights are $1/2,0$. 
If $v$ is false, the signs are exchanged. 
These are precisely the choices used to obtain~\eqref{eq:mask-test-targets}, meaning that the truth value determines the entire pair of tests, and either truth value gives the target distributions for both tests.

\proofstep{IV}{Enforcing the clauses}

For each $C_j$, add a test state $u_j$ and fix $W_j := k_j+L_j-1$. 
Its transitions are fixed as follows: assign weight $1/W_j$ to the literal state for each occurrence in $C_j$, and weight $2^r/W_j$ to $q_{w_{j,r}}^+$ for $0\le r<h_j$. 
If a literal occurs more than once, add its weights in the corresponding transition entry, and all these weights sum to $1$.

Let $r_\ell$ denote the $\mathsf t$-emission probability of the state for literal $\ell$. 
The conditional probability (on the identifying first observation of $u_j$) of $\mathsf t$ is
\begin{equation}
 \frac1{W_j}\left(\sum_{\ell\in C_j}r_\ell+
                  \sum_{r=0}^{h_j-1}2^r r_{w_{j,r}}^+\right).
 \label{eq:mask-clause-probability}
\end{equation}
For the assignment of the target HMM, this equals $\epsilon L_j/W_j$. 
Once~\eqref{eq:mask-boolean} holds, matching this value is exactly equation~\eqref{eq:mask-clause-bits} after dividing by $\epsilon$.

The clause test must also match all other second observations. 
Its probabilities for every $y_v^b$ are fixed, because its transition weights and the $y_v^b$ emission entries are fixed. 
Its probability for $\mathsf g$ is $1-2\epsilon$. 
Its probability for $\mathsf t$ and $\mathsf f$ is in $\{0, \epsilon\}$, and one is $\epsilon$ minus the other by~\eqref{eq:mask-row-sum}. 
All remaining second observations have probability zero. 
Therefore, matching~\eqref{eq:mask-clause-probability} suffices to match the whole distribution for this test.

\proofstep{V}{Correctness of the reduction}

The test states are the two variable tests for every $v\in V$ and the $J$ clause tests, so $R=2|V|+J$. 
We fix their emissions to their distinct identifying symbols, and their initial probabilities to $1/R$. 
All literal and comparison states transition to $q_\#$ with probability $1$. 
Fix $B(q_\#,\#)=A(q_\#,q_\#)=1$. 
These rules, together with the already specified masked entries, complete the HMM and mask. 
The target uses the extended assignment $\sigma^\star$ and the transitions of Step~III. 

Notice that every observation word with positive probability is either 1) the empty word, which has probability 1; 2) a single symbol $a_u$ that identifies a test, which has probability $1/R$; or 3) a word of the form $a_u b \#^t$ with $t \ge 0$, which has the same probability as $a_u b$ in~\eqref{eq:mask-mixture}.
Every other word has probability $0$. 
The tests above therefore characterize equivalence for all observation words, not just the selected two-letter words used to derive the constraints.

Every equivalent masked HMM yields, by Step~III., a Boolean value for every variable in $V$. 
Step~IV.\ then forces the clause equations, so the values of $z,x_1,\ldots,x_N$ satisfy $\psi$. 
Conversely, any satisfying assignment of $\psi$ extends uniquely to the added binary variables. 
Set the literal emissions according to this assignment and use the unique transitions of Step~III. 
The variable tests then have distributions~\eqref{eq:mask-test-targets}, the clause tests match by Step~IV., and the result HMM is equivalent to the target as established earlier in this step.

For a fixed satisfying assignment of $\psi$, the added variables, all open literal emissions, and all open test transitions are uniquely determined. 
There are no other open entries. 
Different assignments change at least one literal emission entry, so they produce different numerical HMMs. 
Consequently, satisfying assignments of $\psi$ correspond bijectively to HMMs in the masked family that are equivalent to the target. 
The target corresponds to $\sigma^\star$, and hence has a distinct equivalent rival exactly when $\varphi$ is satisfiable by~\eqref{eq:mask-dummy}.

Finally, clauses of $\psi$ have at most four literals, so $h_j\le3$ and $L_j\le8$. 
There are linearly many variables, test states, literal states, comparison states, and observation symbols in the size of $\varphi$. 
The fixed probabilities have polynomial bit length, including $1/R$ and the clause weights. 
Thus even listing all matrices and the mask takes polynomial time. 
Every open entry has its own independent parameter---the only domain constraints are nonnegativity and stochasticity.
Repeated uses of a literal refer to the same state, and no two different open entries are required by the mask to share a parameter. 
This concludes the correctness of the reduction.
\end{proof}

\section{HMM permutation-equivalence is GI-complete}\label{app:permutation_GI_complete}

\thmpermequivhard*

\begin{proof}
We show both directions.

For $\gi\le_p$ HMM permutation equivalence, let $G=(V,E)$ be a simple undirected graph on $n$ vertices and set $c:=1/(n+1)$.
Construct an $n$-state HMM with uniform initial distribution, one observation symbol emitted with probability $1$ from every state, and transition matrix
\[
A_G(u,v):=
\begin{cases}
c,&u\ne v\text{ and }\{u,v\}\in E,\\
1-c\,\deg_G(u),&u=v,\\
0,&\text{otherwise}.
\end{cases}
\]
Every row sums to one, and the diagonal entry is strictly larger than $c$ because
$\deg_G(u)\le n-1$.
Hence the off-diagonal entries equal to $c$ recover exactly the edges of $G$.
The initial and emission distributions are invariant under every permutation, so
\[
    G_1\cong G_2
    \quad\Longleftrightarrow\quad
    A_{G_2}=\Pi A_{G_1}\Pi\transpose
    \text{ for some permutation }\Pi
    \quad\Longleftrightarrow\quad
    \model_{G_1}\similarperm\model_{G_2}.
\]

For the converse, given rational HMMs $\model_1,\model_2$, view each hidden state as a vertex labeled by the pair consisting of its initial probability and its emission row, and put a directed edge $u\to v$ labeled by the rational weight $A(u,v)$ whenever that weight is nonzero.
Then HMM permutation equivalence is exactly isomorphism of these vertex-labeled, edge-labeled directed graphs.
Only polynomially many rational labels occur in the input, so replace equal labels by equal finite colors.
Colored directed graph isomorphism reduces in polynomial time to ordinary graph isomorphism~\citep{hsieh2006efficient}.
Thus HMM permutation equivalence $\le_p\gi$.
\end{proof}

\section{HMM identifiability up to state renaming and bisimilarity}
\subsection{HMM identifiability up to bisimilarity is in $\coetr$}
\label{app:bisim-etr}

\thmbisimetr*

\begin{proof}
    Let $(\phi, \Theta)$ be a parameterization and $\theta^\star$ provide the target HMM $\model^\star = (n, \pi^\star, A^\star, B^\star) := \phi(\theta^\star)$.
    We extend the $\etr$ certificate from Theorem~\ref{thm:det-eq-etr-certificate}.
    As before, we existentially guess a rival $\model' = (n, \pi', A', B')$ through matrices $(F,\{\vv*{T}{y}\})$, and impose constraints (i)--(iii) from Theorem~\ref{thm:det-eq-etr-certificate}.
    It remains to encode that $\model'$ is not classically bisimilar to $\target$.

    Let the target states be $1,\ldots,n$ and the rival states be $n+1,\ldots,2n$.
    We encode a partition-refinement sequence \citep{paige1987three} by Boolean variables
    \[
        P^t_{ij}\in\{0,1\},
        \qquad
        t=0,\ldots,2n,\quad i,j\in\{1,\ldots,2n\},
    \]
    where $P^t_{ij}=1$ means that states $i,j$ are in the same block of the $t$-th partition.
    First, we impose that each $P^t$ is an equivalence relation:
    \[
        P^t_{ij}(P^t_{ij}-1)=0,\qquad
        P^t_{ii}=1,\qquad
        P^t_{ij}=P^t_{ji},\qquad
        P^t_{ij}P^t_{jk}(1-P^t_{ik})=0.
    \]

    The initial partition groups together states if and only if they have the same observation distribution, which we encode by, for $i,j\in\{1,\ldots,2n\}$, introducing constraint
    \[
        P^0_{ij}=1
        \;\longleftrightarrow\;
        \sum_{y=1}^{m}
        \bigl(
            \widehat B(i,y)-\widehat B(j,y)
        \bigr)^2=0,
    \]
    where
    \[
        \widehat B(i,y)=
        \begin{cases}
            B^\star(i,y), & i\le n,\\
            \obs'(i-n,y), & i>n.
        \end{cases}
    \]

    The refinement step says that two states remain together iff they were
    together before and have the same transition mass into every current block.
    Thus, for every $t<2n$ and every $i,j\in\{1,\ldots,2n\}$, we introduce the following constraint:
    \[
        P^{t+1}_{ij}=1
        \;\longleftrightarrow\;
        \left(
            P^t_{ij}=1
            \;\wedge\;
            \bigwedge_{k=1}^{2n}
            \left[
                \sum_{u=1}^{2n} P^t_{ku}\widehat A(i,u)
                =
                \sum_{u=1}^{2n} P^t_{ku}\widehat A(j,u)
            \right]
        \right),
    \]
    where
    \[
        \widehat A(i,u)=
        \begin{cases}
            A^\star(i,u), & i\le n,\ u\le n,\\
            \trans'(i-n,u-n), & i>n,\ u>n,\\
            0, & \text{otherwise}.
        \end{cases}
    \]

    Finally, to enforce that the models are not bisimilar, $\target \not \sim \rival$, require that the initial distributions disagree on some final
    block:
    \[
        \sum_{k=1}^{2n}
        \left(
            \sum_{u=1}^{n}P^{2n}_{ku}\init^\star(u)
            -
            \sum_{u=1}^{n}P^{2n}_{k,n+u}\init'(u)
        \right)^2>0.
    \]

    The formula above describes the partition-refinement algorithm for classical HMM bisimulation~\citep{paige1987three,larsen1989bisimulation}.
    Since there are only $2n$ states, after at most $2n$ iterations the partition is stable and the algorithm can terminate (i.e., we have reached a fixpoint), meaning that we have reached the coarsest bisimulation. 
    We refer to \citet{paige1987three} and \citet{larsen1989bisimulation} for a proof of this standard fact.
    The final constraint says that the two initial distributions assign different mass to some block, so $\target$ and $\model'$ are not classically bisimilar.

    The certificate has polynomial size and all constraints are polynomial constraints and finite Boolean combinations thereof.
    We recall that these Boolean combinations can be done in $\etr$, as explained in Section~\ref{sec:etr}.
    Hence $\coiddetbisim\in\etr$.
\end{proof}

\subsection{$\coetr$-hardness of deterministic identifiability up to bisimulation with simple parameters}
\label{app:bisim_simple_hardness}

\corbisimetrhard*

\begin{proof}
Take the PA family constructed in Theorem~\ref{thm:pa_identifiability_hard}.
Its target has $z=0$, and every distinct output-equivalent rival has $z=1$.
We extend the family by preventing the rivals from being bisimilar to the target while preserving output equivalence.
We use the well-known fact that making a probabilistic choice earlier or later can preserve output probabilities without preserving classical bisimilarity (see for example \citet[Figure~1]{doyen2008equivalence}).
We specify these branches next.

Introduce a fresh letter $\sigma_{\not\sim}$ and five fresh nonaccepting states
$q_1^{\not\sim},\dots,q_5^{\not\sim}$, with transitions
\[
 q_{\mathrm{init}}
 \xrightarrow[\sigma_{\not\sim}]{1-z}q_1^{\not\sim},
 \qquad
 q_{\mathrm{init}}
 \xrightarrow[\sigma_{\not\sim}]{z}q_2^{\not\sim},
\]
and
\[
\begin{aligned}
 q_1^{\not\sim}&\xrightarrow[t]{1/2}q_3^{\not\sim},
 &q_1^{\not\sim}&\xrightarrow[t]{1/2}q_4^{\not\sim},
 &q_2^{\not\sim}&\xrightarrow[t]{1}q_5^{\not\sim},
 \\[2pt]
 q_3^{\not\sim}&\xrightarrow[t]{1}q_{\mathrm{acc}},
 &q_4^{\not\sim}&\xrightarrow[t]{1}q_{\mathrm{rej}},
 &q_5^{\not\sim}&\xrightarrow[t]{1/2}q_{\mathrm{acc}},
 &q_5^{\not\sim}&\xrightarrow[t]{1/2}q_{\mathrm{rej}}.
\end{aligned}
\]
Every letter other than $t$ leads from each new state to
$q_{\mathrm{rej}}$ with probability $1$.
The new letter $\sigma_{\not\sim}$ also leads to $q_{\mathrm{rej}}$
from every old state (other than $q_{\mathrm{init}}$).

From either $q_1^{\not\sim}$ or $q_2^{\not\sim}$, the only accepted word is $t^2$, with probability $1/2$.
Starting from $q_1^{\not\sim}$, the choice between acceptance and rejection is made on the first $t$.
Starting from $q_2^{\not\sim}$, it is made on the second $t$.
Thus the only new accepted word is $\sigma_{\not\sim}t^2$, and its acceptance probability is
\[
 (1-z)\tfrac12+z\tfrac12=\tfrac12.
\]
Note that all original acceptance probabilities remain unchanged, which means that the extended PA has exactly the same output-equivalent parameter values as before.

Apply the PA-to-HMM transformation from Appendix~\ref{app:hmm_identifiability_simple}, and suppose that a rival with $z=1$ were bisimilar to the target.
Then, their initial control states must be related, and each has a unique successor emitting $\sigma_{\not\sim}$, so these
letter states must also be related.
Their deterministic transitions then force the target's $q_{1,\mathrm c}^{\not\sim}$ to be related to the rival's $q_{2,\mathrm c}^{\not\sim}$.
Likewise, their unique successors emitting $t$ must be related.
However, the target's $q_{1,t}^{\not\sim}$ moves with equal probability to $q_{3,\mathrm c}^{\not\sim}$ and $q_{4,\mathrm c}^{\not\sim}$, whereas the rival's $q_{2,t}^{\not\sim}$ moves deterministically to $q_{5,\mathrm c}^{\not\sim}$.
From these three successor states, the probabilities of observing $\#t\#\cmark$ are, respectively, $\kappa^2$, $0$, and $\tfrac12 \kappa^2$, where $\kappa>0$ is the fixed probability of choosing a particular
letter or the decision step.
Bisimilar states have equal observation probabilities, so neither target successor can belong to the bisimulation class of the rival successor.
The two bisimilar letter states therefore assign probabilities $0$ and $1$ to that class, a contradiction.
\end{proof}

\subsection{$\conp$-hardness of deterministic identifiability up to bisimulation with masks}
\label{app:bisim_mask_hardness}

\corbisimnphard*

\begin{proof}
We use the reduction of Theorem~\ref{thm:iddeteq_mask_nphard} without modification. 
Recall that HMMs that are output-equivalent to the target correspond bijectively to satisfying assignments of the extended Boolean formula. 
The target is the unique such model with $z=0$, and every distinct equivalent rival has $z=1$. 
A distinct equivalent rival exists exactly when the input $\mathsf{3SAT}$ formula is satisfiable.

Fix such a rival. In the notation of Appendix~\ref{app:mask_nphard}, its positive-literal state $q_z^+$ has emission probabilities
\[
 B(q_z^+,y_z^+)=\epsilon,\qquad
 B(q_z^+,\mathsf g)=1-2\epsilon,\qquad
 B(q_z^+,\mathsf t)=\epsilon,\qquad
 B(q_z^+,\mathsf f)=0,
 \qquad \epsilon=1/4.
\]
No state of the target has this emission distribution as we argue next. 
The only target states emitting $y_z^+$ are $q_z^+$ and the two comparison states $q_{z,+}^{\mathsf T},q_{z,+}^{\mathsf F}$. 
The target's $q_z^+$ emits $\mathsf t$ with probability $0$, whereas each comparison state emits its additional symbol $c_z^{\mathsf T}$ or $c_z^{\mathsf F}$ with probability $\epsilon$. 
Both of these additional symbols have probability $0$ at the rival's $q_z^+$.
Suppose the rival and target were bisimilar. 
Their copies of the true-literal test state $u_z^{\mathsf T}$ would have to be related: its identifying emission symbol is used by no other state, and its initial probability is $1/R>0$ in each model. 
The rival sends probability $1/2$ from this test state to $q_z^+$, because $z$ is true. 
But the bisimulation class containing the rival's $q_z^+$ contains no target state, since bisimulation requires identical emission distributions. 
The rival therefore sends at least $1/2$ into this class, while the target sends $0$, a contradiction.

Every distinct equivalent rival is consequently non-bisimilar to the target. 
A satisfiable input yields such a rival, and an unsatisfiable input leaves only the target. 
\end{proof}

\section{The hardness of HMM identifiability is independent of alphabet size}
\label{app:alphabet_independent}

\binalphabet*

\begin{proof}
The reduction from $\biniddeteq$ to $\iddeteq$ is immediate, so we prove the converse direction.
Let $\model=(n,\pi,A,B)$ be an HMM over $\abc=\{1,\ldots,m\}$ with parameterization $\phi$.
Let $b:=\lceil\log_2m\rceil$ and fix an injective encoding $c:\abc\to\{0,1\}^b$.

We replace every original hidden state $i$ by a small binary-output module.
It contains a \emph{choice state} $q_i$, an \emph{exit state} $e_i$, and, for every symbol $y\in\abc$, a deterministic code chain $q_{i,y,1},\ldots,q_{i,y,b}$.
All emissions in the new HMM are deterministic: $q_i$ and $e_i$ emit $0$, while $q_{i,y,j}$ emits the $j$-th bit of $c(y)$.
The initial distribution is
\[
    \widehat\pi(q_i)=\pi_i
\]
and is zero on all other new states.
The transitions are
\[
\widehat A(q_i,q_{i,y,1})=B_{iy},
\]
\[
\widehat A(q_{i,y,j},q_{i,y,j+1})=1
\quad(j<b),
\qquad
\widehat A(q_{i,y,b},e_i)=1,
\]
and
\[
\widehat A(e_i,q_{i'})=A_{ii'}.
\]
All unspecified transitions are zero.
Observe that every row is stochastic.

Thus one original observation $y$ produced in hidden state $i$ is replaced by the fixed binary block
\[
    0\,c(y)\,0,
\]
and the original transition from $i$ to the next hidden state occurs only after that block has been emitted.
Consequently, if
\[
E(y_1\cdots y_\ell)
:=
0c(y_1)00c(y_2)0\cdots 0c(y_\ell)0,
\]
then
\[
    \prob_{\widehat\model}(E(w))=\prob_\model(w)
    \qquad\text{for every }w\in\abc^*,
\]
while every binary word that is not a prefix of a valid block encoding has probability zero.
If a binary word ends partway through a valid block, its probability is the sum of the probabilities of the original words obtained by completing that block with each symbol whose code has the observed binary prefix.
Thus every binary-word probability is determined by the original probabilities of observation words.
Conversely, the complete encodings $E(w)$ recover each original word probability exactly.
Hence two original HMMs are output equivalent iff their binary encodings are output equivalent.

The construction also preserves the simple and mask parameter restrictions used in this paper.
Every original initial entry $\pi_i$ appears once as $\widehat\pi(q_i)$,
every emission entry $B_{iy}$ appears once as the transition
$q_i\to q_{i,y,1}$, and every transition entry $A_{ii'}$ appears once as the transition
$e_i\to q_{i'}$.
All other new entries are fixed to $0$ or $1$.
Therefore a simple parameterization remains simple, and a mask remains a mask, with exactly the same parameters.
The map from original rivals to binary rivals is bijective within the induced parameterized family and preserves strict/literal equality of HMM parameters.
\end{proof}

\section{Deterministic local identifiability is in $\coetr$}
\label{app:local_etr}

\thmsmallradius*

\begin{proof}
Define
\begin{equation}\label{eq:f-target}
    f_\target(\rival)
    :=
    \sum_{w\in\abc^{\le 2n-1}}
    \bigl(\prob_\target(w)-\prob_\rival(w)\bigr)^2
    =
    \sum_{w\in\abc^{\le 2n-1}}
    \bigl(\one\transpose T^\star(w)-\one\transpose T'(w)\bigr)^2,
\end{equation}
where the second equality is given by Remark~\ref{rem:forward_backward}.  
By \citet{paz1971introduction}, for $n$-state rivals we have that $f_\target(\rival)=0$ if and only if $\target\equiv\rival$ (see the introductory discussion in Section~\ref{sec:hmm_equivalence_certificate}).
Thus, local identifiability means exactly that $\target$ is the only HMM in some ball around $\target$ on which $f_\target$ returns 0.

Let $K := n+n^2+nm$ be the number of entries in $(\pi,A,B)$, and view rivals as points in $[0,1]^K$.
Since the entries of $\target$ are rational, we can choose a positive integer $D$ divisible by all the denominators.
For $|w|\le 2n-1$, the probability $\prob_\target(w)$ is a sum of products of at most $4n-2$ entries of $\target$, so it is a rational number whose denominator divides $D^{4n}$.
Hence $\widehat f_\target := D^{8n}f_\target$ is a polynomial with integer coefficients, and multiplying by $D^{8n}$ does not change the zero set.
We use $\widehat f_\target$ in place of $f_\target$ from now on.
Similarly, we replace the squared distance
$\|\target-\rival\|_2^2$ by the integer polynomial
\[
    \widehat g(\rival):=D^2\|\target-\rival\|_2^2.
\]

Consider the following semialgebraic set that includes all rivals that are output-equivalent to $\target$:
\begin{equation}\label{eq:semialgebraic-set-equivalence}
    Z_\target :=
    \{\rival\in[0,1]^K \mid \rival \text{ is stochastic and } \widehat f_\target(\rival)=0\}.
\end{equation}
The set $Z_\target$ is closed and bounded in $\mathbb R^K$, hence compact (Heine--Borel theorem), and it is defined by the polynomial equation $\widehat f_\target(\rival) = 0$ and non-strict inequalities concerning stochasticity.

We now fix the parameters needed to apply the lower bound of \citet[Theorem~1]{jeronimo2013minimum} to $\widehat g$ on $Z_\target$.
\begin{itemize}
    \item \emph{Degree.}
    A word probability of length $\ell$ has degree $2\ell$ in the rival entries, so $\widehat f_\target$ has degree at most $4(2n-1)$, and $\widehat g$ has degree $2$.
    We therefore take the even upper bound $d:=8n$.
    \item \emph{Coefficients.}
    There are at most $m^{2n}$ words of length $\le 2n-1$, and squaring a word probability produces at most $n^{4n}$ products of pairs of hidden paths.
    Thus, very crudely, the integer coefficients of $\widehat f_\target$ are bounded in absolute value by $4D^{8n}m^{2n}n^{4n}$.
    The coefficients of $\widehat g$ are bounded by $(2n+1)D^2$, where the largest one is the constant term $D^2\|\target\|_2^2$, and $\|\target\|_2^2\le 2n+1$ because every stochastic vector has squared norm at most $1$.
    We take
    \[
        H:=\max\{4D^{8n}m^{2n}n^{4n},\,(2n+1)D^2\}.
    \]
    \item \emph{Number of polynomials.}
    Let $s$ be the number of polynomials used to define $Z_\target$: one for the equality $\widehat f_\target=0$, $2n+1$ for the stochasticity equalities, and $2K$ for the box inequalities.
    That is, $s:=2K+2n+2$.
\end{itemize}
Finally, let $\widetilde H := \max\{H,2K+2s\}$.
Although $H$ may be numerically enormous, its binary length is polynomial in the input size, and it can be computed without expanding the exponentially large polynomial $\widehat f_\target$.

Assume now that $\target$ is locally identifiable.
Then the connected component of $Z_\target$ containing $\target$ is the singleton $\{\target\}$.  
Here a connected component means a maximal connected subset of $Z_\target$.

Let $C$ be any other connected component of $Z_\target$.
Since $Z_\target$ is compact and connected components of compact sets are closed, $C$ is also compact.
Since $C$ does not contain $\target$, the polynomial $\widehat g$ has a positive minimum on $C$.

\citet[Theorem~1]{jeronimo2013minimum} applies to a polynomial with integer coefficients minimized over a compact connected component of a set
defined by polynomial equations and non-strict inequalities.
The theorem gives us that
\begin{equation}\label{eq:jpt-bound}
    \min_{\rival\in C} \widehat g(\rival)
    \ge
    \left(2^{4-K/2}\widetilde H d^K\right)^{-E},
    \qquad
    E := K2^K d^K .
\end{equation}
Let $A_0 := 16\,\widetilde H d^K$.
Since $2^{4-K/2}\widetilde H d^K \le A_0$, and dividing by $D^2$ undoes the scaling in $\widehat g$, \eqref{eq:jpt-bound} gives
\[
    \|\target-\rival\|_2^2
    \ge
    \frac{1}{D^2A_0^E}
\]
for every $\rival\in C$.
This lower bound is independent of the component $C$.  

Let
\[
    M_\target := 1 + 2\lceil\log_2 D\rceil + E\lceil\log_2 A_0\rceil,
    \qquad
    L_\target := \lceil\log_2 M_\target\rceil,
\]
Then
\[
    2^{-2^{L_\target}}
    \le 2^{-M_\target}
    <
    \frac{1}{D^2A_0^E}
\]
because $D^2A_0^E \le 2^{M_\target-1}$.
Then, if $\target$ is locally identifiable, no distinct equivalent rival HMM lies within distance $2^{-2^{L_\target}}$ of $\target$.

Finally, $L_\target$ is polynomially bounded and polynomial-time computable:
$K$ and $d$ are polynomial in the HMM dimensions, while
$\log D$, $\log H$, and therefore the binary lengths of $A_0$, $E$, and $M_\target$ are polynomial in the input size.
Only $L_\target$, not the doubly exponentially small radius itself, will be written explicitly in the next certificate.
\end{proof}

\section{Generic identifiability is in $\pspace$}
\label{app:generic-upper-bound}

\begin{proof}
We follow the proof of \citet[Theorem~11]{dorfler2024complexity}, using our output-equivalence certificate to describe the non-identifiable set $\mathsf{Bad}_\phi$ in \eqref{eq:generic-bad-set}.
Let formula $\varphi(\theta)$ define $\Theta$, and let $\psi(\theta,\theta', F, \{\vv*{T}{y}\}_{y\in\abc})$ be the polynomial-size formula for checking (i)--(iii) from Theorem~\ref{thm:det-eq-etr-certificate}, where $F$ and $\{\vv*{T}{y}\}_{y\in\abc}$ are the certificate matrices. 
Then membership in $\mathsf{Bad}_{\phi}$ is expressed by
\begin{equation}
 \varphi(\theta)\land\exists\theta',F, \{\vv*{T}{y}\}_{y\in\abc}, \,
 \bigl(\varphi(\theta')\land \psi(\theta,\theta',F, \{\vv*{T}{y}\}_{y\in\abc})
       \land\phi(\theta')\ne\phi(\theta)\bigr).
 \label{eq:generic-existential-description}
\end{equation}
The last condition means that at least one entry of the two numerical HMMs differs, so it is a polynomial-size disjunction of polynomial inequalities. 
Thus both $\Theta$ and $\mathsf{Bad}_{\phi}$ have polynomial-size $\etr$ descriptions. 
Only $\theta$ remains free in~\eqref{eq:generic-existential-description}, on which we will compute the algebraic dimension.

Koiran's algorithm reduces the question $\dim S\ge r$ to the existential theory of the reals in polynomial time\footnote{although $\etr$ had not yet been established as a complexity class at the time of the algorithm's publication.}, also when $S$ is given by an existential formula~\citep[Section~1.1 and Theorem~6]{koiran1999real}; see also \citet[Lemma~9 and Appendix~D]{dorfler2024complexity}. Assuming oracle access to the existential theory of the reals, we can therefore decide generic identifiability as follows:
\begin{enumerate}
 \item Compute $\dim\Theta$ by applying Koiran's algorithm to test $\dim\Theta\ge r$ for $r=0,\ldots,K$.
 \item Compute $\dim\mathsf{Bad}_{\phi}$ by testing $\dim\mathsf{Bad}_{\phi}\ge r$ for the same values of $r$.
 \item Accept if $\Theta=\varnothing$ or $\dim\mathsf{Bad}_{\phi}<\dim\Theta$, otherwise reject.
\end{enumerate}
In each of the first two steps, the dimension is the largest successful threshold, or $-1$ if none succeeds. The procedure is correct by the definition of generic identifiability. It uses polynomially many existential real queries and hence already gives a polynomial-space algorithm by \citet{canny1988some} and \citet{renegar1992computational}.

To obtain the sharper bound, replace these iterations by the single formula
\begin{equation}
 \neg(\dim\Theta\ge0)
 \ \lor\
 \bigvee_{r=0}^{K}
 \bigl((\dim\Theta\ge r)
       \land\neg(\dim\mathsf{Bad}_{\phi}\ge r)\bigr).
 \label{eq:generic-threshold-disjunction}
\end{equation}
The first disjunct handles an empty domain. 
For a nonempty domain, a successful threshold gives $\dim\mathsf{Bad}_{\phi}<r\le\dim\Theta$. 
Conversely, non-identifiability makes the disjunct with $r=\dim\Theta$ true.

Each dimension check in~\eqref{eq:generic-threshold-disjunction} has a polynomial-size existentially quantified formula by Koiran's algorithm.
To bring the quantification to the outer part of the formula, we use a new set of fresh variables per check.
Since the right-hand side of the conjunction in~\eqref{eq:generic-threshold-disjunction} is negated, it becomes universally quantified. 

Within each conjunction, the existential and universal formulas have no variables in common, so their quantifiers can be placed in either order. 
We thus obtain both an $\exists\forall$ sentence and a $\forall\exists$ sentence, each constructible in polynomial time. 
This proves
\[
 \idgeneric\in\exists\forall\mathbb R\cap\forall\exists\mathbb R \subseteq \pspace,
\]
where the $\pspace$ bound is given by \citet{canny1988some} and \citet{renegar1992computational}.
\end{proof}

\section{Tune--Nguyen--Roughan criterion}\label{app:tnrcondition}

This section considers the sound and complete criterion provided by \citet{tune2013hidden}.
Their criterion makes use of the row-wise tensor product and Kruskal rank as we introduce next.

\begin{definition}[Row-wise tensor product and Kruskal rank]
For matrices $U\in\real^{q\times r}$ and $V\in\real^{q\times s}$ with the same number of rows,
\[
U\rowtensor V
:=
\begin{pmatrix}
U_{1,*}\otimes V_{1,*}\\
\vdots\\
U_{q,*}\otimes V_{q,*}
\end{pmatrix}.
\]
The Kruskal rank $\krank(M)$ is the largest $k$ such that every set of $k$ rows of $M$ is linearly independent.
\end{definition}

\begin{theorem}[\citet{tune2013hidden}, Theorem~5]\label{thm:tnr}
Given an HMM $\model = (\pi, A, B)$ with an irreducible and aperiodic Markov chain and a finite non-redundant alphabet,
the HMM is deterministically globally identifiable up to hidden-state permutation and scaling if and only if
\[
    \krank(B\rowtensor A)=n.
\]
\end{theorem}

The matrix in Theorem~\ref{thm:tnr} has exactly $n$ rows, meaning that its Kruskal rank is $n$ precisely when its set of all $n$ rows is linearly independent, which is exactly full row rank.
Gaussian elimination checks this in polynomial time.

To relate their setting to ours, the models they consider are parameterized with fully open rectangular masks, that is, 
$\Theta$ is the set $[0, 1]^{n + n^2 + n m}$ restricted to the stochasticity constraints of HMMs and $\phi$ maps each parameter one-to-one to each HMM entry.
Therefore, the rival's parameters are fully unrestricted, and only specific degenerate target HMMs, those with a non-full Kruskal rank, are unidentifiable in the sense that there exist output-equivalent rivals that are not equal to the target up to state permutation and scaling.

\end{document}